\documentclass[10pt,twocolumn]{IEEEtran}
\usepackage{verbatim}
\usepackage{amssymb}  
\usepackage{color}
\usepackage{bm}
\usepackage{xcolor}
\usepackage{algorithm}
\usepackage{algorithmic}
\usepackage{theorem}  
\usepackage{booktabs}

\newcommand{\tabincell}[2]{\begin{tabular}{@{}#1@{}}#2\end{tabular}}

\newcommand{\secref}[1]{Sec. \ref{#1}}

\usepackage{cite}

\ifCLASSINFOpdf
  \usepackage[pdftex]{graphicx}
\else
  \usepackage[dvips]{graphicx}
\fi
\usepackage{amsmath}

\usepackage{algorithmic}

\usepackage{array}

\ifCLASSOPTIONcompsoc
  \usepackage[caption=false,font=normalsize,labelfont=sf,textfont=sf]{subfig}
\else
  \usepackage[caption=false,font=footnotesize]{subfig}
\fi

\usepackage{stfloats}

\usepackage{url}

\theoremstyle{plain}

\theoremstyle{plain}
\newtheorem{theorem}{Theorem}
\theoremstyle{plain}
\newtheorem{proposition}{Proposition}
\theoremstyle{plain}
\newtheorem{corollary}{Corollary}
\theoremstyle{plain}
\newtheorem{assumption}{Assumption}
\theoremstyle{plain}
\newtheorem{lemma}{Lemma}

\begin{document}

\def\QEDclosed{\mbox{\rule[0pt]{1.3ex}{1.3ex}}}
\def\QEDopen{{\setlength{\fboxsep}{0pt}\setlength{\fboxrule}{0.2pt}\fbox{\rule[0pt]{0pt}{1.3ex}\rule[0pt]{1.3ex}{0pt}}}}
\def\QED{\QEDopen}
\def\proof{}
\def\endproof{\hspace*{\fill}~\QED\par\endtrivlist\unskip}
%

\title{Low-Complexity Eigenvector Prediction-based Precoding Matrix Prediction in Massive MIMO with Mobility}
%
%
%

\author{Ziao~Qin, Haifan~Yin,~\IEEEmembership{Senior Member,~IEEE}, and Weidong~Li 
\thanks{Z. Qin, H. Yin and W. Li are with School of Electronic Information and Communications, Huazhong University of Science and Technology, 430074 Wuhan, China (e-mail: ziao\_qin@hust.edu.cn, yin@hust.edu.cn, weidongli@hust.edu.cn).}

\thanks{This work was supported by the National Natural Science Foundation of China under Grant 62071191. The corresponding author is Haifan Yin.}}

\maketitle

\begin{abstract}
In practical massive multiple-input multiple-output (MIMO) systems, the precoding matrix is often obtained from the eigenvectors of channel matrices and is challenging to update in time due to finite computation resources at the base station, especially in mobile scenarios. In order to reduce the precoding complexity while enhancing the spectral efficiency (SE), a novel precoding matrix prediction method based on the eigenvector prediction (EGVP) is proposed. The basic idea is to decompose the periodic uplink channel eigenvector samples into a linear combination of the channel state information (CSI) and channel weights. We further prove that the channel weights can be interpolated by an exponential model corresponding to the Doppler characteristics of the CSI. A fast matrix pencil prediction (FMPP) method is also devised to predict the CSI. We also prove that our scheme achieves asymptotically error-free precoder prediction with a distinct complexity advantage. Simulation results show that under the perfect non-delayed CSI, the proposed EGVP method reduces floating point operations by 80\% without losing SE performance compared to the traditional full-time precoding scheme. In more realistic cases with CSI delays, the proposed EGVP-FMPP scheme has clear SE performance gains compared to the precoding scheme widely used in current communication systems.
\end{abstract}
\begin{IEEEkeywords}
massive MIMO, CSI, mobility, precoding matrix prediction, channel prediction.
\end{IEEEkeywords}

%
\IEEEpeerreviewmaketitle

\section{Introduction}\label{sec1}
%
%
%
%
\IEEEPARstart{M}{a}ssive multiple-input multiple-output (MIMO) is an integral part of the fifth generation (5G) new radio (NR) mobile communication systems due to its outstanding performance in terms of spectral efficiency (SE) and energy efficiency \cite{2010Marzetta,2013NgoEnergy}. One of the key feature of massive MIMO is the large number of antennas at the base station (BS). A sufficient number of BS antennas enables linear precoders at the BS, such as zero-forcing (ZF), to achieve near-optimal SE, as well as error-free channel estimation and perfect noise cancellation  \cite{2013HoydisHowmany}. In order to achieve high spatial multiplexing gains, the downlink precoders are calculated based on the prior-estimated downlink channel state information (CSI). 

The CSI acquisition dilemma in massive MIMO systems with non-negligible mobility has recently attracted considerable attention in both academia \cite{2015ZhouMobility,2016MoblitySurvey} and industry \cite{2019mobilityReport,2022mobilityMeeting}. 
Usually, channel prediction is an effective way to help the BS to timely update CSI in the fast time-varying massive MIMO systems. In this scenario, additional methods should be introduced to mitigate the influence of high mobility on CSI acquisition \cite{2016MoblitySurvey}. Recently, many research works have proposed effective channel prediction methods. By exploiting the specific angle-delay-Doppler structure of the multipath, a Prony-based angular-delay domain (PAD) channel prediction was proposed in \cite{2020yinMobility}. For an orthogonal time frequency system (OTFS) system, the channel prediction was achieved by a path division multiple access (PDMA) in \cite{2021OTFS}. The authors of \cite{2022JSACParallel} proposed a channel prediction method based on the attention mechanism in machine learning to avoid propagation error. These works leveraged the channel prediction to update the CSI through a polynomial framework \cite{2020yinMobility} or non-polynomial algorithm \cite{2022JSACParallel}. 

In addition to the dilemma of timely CSI acquisition, the overwhelming precoding complexity is one of the major challenges for the practical massive MIMO systems. In current mobile communication systems, the UEs are equipped with two or more antennas \cite{3gpp214} and eigen zero-forcing (EZF) \cite{2010TSPeigenvalue} is the \textit{de-facto} method for precoding \cite{20175Gtrial,2021trial}. The EZF precoding scheme inevitability introduces an eigenvector decomposition (EVD) or singular value decomposition (SVD) of the CSI to select a particular eigenvector as the precoder for the UEs. The large number of BS antennas results in a high dimensional SVD operation. There are some existing works addressing the precoding complexity reduction in massive MIMO. The work \cite{2022DoAEigen} proposed two low-complexity eigenvector-based estimators utilizing partitioned sub-array auto-correlation and cross-correlation combining. The authors in \cite{2020gramesti} exploited cross-carrier correlation to interpolate the Gram matrix for precoding and reduce the complexity. And the authors in \cite{2021flag} reduced the precoding complexity by interpolating the CSI with a flag-manifold-based method.

In practical 5G systems with mobility, the timely precoding is more challenging. Due to the limited computational ability, the BS often updates the precoders in a periodic way \cite{3gpp214} and it is hard to update the precoders in every subframe even though the CSI may be predicted  \cite{2020yinMobility,2021OTFS,2022JSACParallel}. Therefore, the major challenge lies in the timely and accurate precoding matrix prediction with low-complexity.

To track the fast time-varying channel while reducing the precoding complexity, we utilize several periodic eigenvectors of the CSI to predict the rest of the precoders, called the eigenvector prediction (EGVP) method. The key idea of EVGP is to decompose the precoder into a linear combination of the channel weights and channels. On the one hand, we prove that the channel weights can be estimated with a complex exponential model, which enables a polynomial-complexity channel weight interpolation by the periodic eigenvector samples. On the other hand, the downlink channel prediction is based on an uplink channel parameter prediction with several uplink channel estimation samples and the channel reciprocity in the time division duplex (TDD) system. As a result, the precoding matrix is predicted with polynomial complexity and the CSI delay can be handled at the same time.

The proposed EGVP-based prediction method is different from the traditional channel prediction schemes \cite{2020yinMobility,2021OTFS,2022JSACParallel} and  precoders complexity reduction schemes \cite{2022DoAEigen,2020gramesti,2021flag}. The traditional channel prediction methods like \cite{2020yinMobility,2021OTFS,2022JSACParallel} often assumed a full-time precoder updating scheme, which is difficult to implement in the practical system with limited computational capability.
The studies in \cite{2022DoAEigen,2020gramesti,2021flag} reduced the precoding complexity in one way or another, while failing to consider the fast-varying channel in mobile environments. To the best of our knowledge, this is the first paper to predict the precoders through a low-complexity eigenvector prediction method in a massive MIMO system with mobility. The main contributions of our paper are summarized below:


\begin{itemize}
    \item We predict the precoding matrix in a massive MIMO system by a low-complexity channel eigenvector prediction, called EGVP method. A fast matrix pencil prediction-based EGVP (EGVP-FMPP) method is further proposed to deal with the CSI delay. The two schemes both show great advantages in terms of the precoding complexity and SE, especially in practical systems with limited SVD capability at the BS.
    \item We decompose the precoding matrix prediction problem into two sub-problems: the channel weight interpolation problem and the channel prediction problem. We further prove that the channel weight can be estimated by a complex exponential model associated with the Doppler characteristics of the channels. The relationship between the periodic length of the SVD of the channel matrix and the maximum speed of the UEs is derived to achieve a non-overlapping channel eigenvector sampling in EGVP.
    \item We give asymptotic analyses of the prediction error (PE) performances of the proposed schemes. The results show that two uplink channel samples and twice the largest number of non-orthogonal paths among all downlink channels may lead to an asymptotically error-free channel prediction and eigenvector interpolation, respectively. Moreover, the numerical results show that EGVP method still functions well in practical systems with limited number of antennas and bandwidth.
    \item We demonstrate the performance advantages of the EGVP method over the benchmark schemes, such as the SE, PE and complexity of precoding. For a common system parameter configuration, the EGVP scheme can reduce the complexity of the traditional full-time SVD-based precoding scheme by 80\% without lossing SE performance. According to the numerical results under an industrial channel model, the SE improvements of the EGVP-FMPP scheme over the benchmarks vary from 20.6\% to 49.1\%.
\end{itemize}

The rest of the paper is divided into six sections. \secref{sec2} formulates the downlink channel and the corresponding eigenvector. \secref{sec3} proposes EGVP algorithm given perfect uplink CSI and \secref{sec4} introduces EGVP-FMPP method to combat the CSI delay. \secref{sec5} analyzes the asymptotic performance and the complexity of our schemes. \secref{sec6} shows the numerical results of our proposed methods in different scenarios. \secref{sec7} concludes this paper.

The bold symbol stands for a matrix or a vector. The matrix dimension is denoted by ${\mathbb{C}^{a \times b}}$, where $a$ is the number of rows and $b$ is the number of columns. The Kronecker product symbol is $ \otimes$. The Moore-Penrose inversion, conjugate transpose, conjugation, transpose of the matrix are denoted by the superscripts of $\dag $, $H$, $*$, and $T$, respectively. The absolute value of a complex number is $\left| x \right|$. The expectation of a random variable $x$ over the time and the frequency is $\mathbb{E}\left\{ {{x}} \right\}$. The inner product of two vectors is $\left\langle {{\bf{x}},{\bf{y}}} \right\rangle$. The convergence to infinity of a variable is denoted by $x \to \infty$. The modulo operation is abbreviated as $\bmod \left( {x,y} \right)$. 
\section{Formulation of CSI and eigenvector}\label{sec2}
We consider a TDD massive MIMO system with $N_t$ antennas at the BS. The downlink channel can be estimated through uplink sounding reference signals (SRS) and the exploitation of channel reciprocity. The $K$ UEs are randomly distributed in the cell and each UE is equipped with $M$ antennas. The number of UE antennas is usually limited in practical communication systems, $M\le4$ . The $k$-th UE moves at a certain speed $v_k$. The wideband system has $N_f$ subcarriers and the center frequency is $f_0$. The duration of a subframe is $\Delta _t$ and  
$\Delta_f$ is the subcarrier spacing. 
\subsection{Channel model}
The channel sparsity in spatial and frequency domain can be exploited by projecting the wideband channel into the angular-delay domain \cite{2020yinMobility,3gpp214}. As a result, we consider the wideband channel in this paper. However, the generalization of our methods to the narrow-band case is straightforward. 

Let the collection of the wideband downlink channel at all subcarriers of the $k$-th UE be 
\begin{equation}\label{full dimension channel}
{{\bf{H}}_k}\left( t \right) = \left[ {\begin{array}{*{20}{c}}
{{{\bf{h}}_{1,k}}\left( t \right)}&{{{\bf{h}}_{2,k}}\left( t \right)}& \cdots &{{{\bf{h}}_{M,k}}\left( t \right)}
\end{array}} \right].
\end{equation}
In the following, we drop the subscript $k$ for simplicity. Each column ${\bf{h}}_m \in \mathbb{C}^{N_fN_t\times1}$ represents the downlink channel of the $m$-th antenna at all subcarriers. Assume a uniform planar array (UPA) at the BS \cite{3gpp901}, the channel ${\bf{h}}_m$ is modeled as
\begin{equation}\label{wideband channel}
    {{\bf{h}}_m}\left( t \right) = \sum\limits_{p = 1}^P {{\beta _{p,m}}{e^{j {\omega _{p,m}}t}}{{\bf{d}}_m}\left( {{\theta _{p,m}},{\varphi _{p,m}},{\tau _{p,m}}} \right)},
\end{equation}
where ${{\bf{d}}_m}\left( {{\theta _{p,m}},{\varphi _{p,m}},{\tau _{p,m}}} \right) = {{\boldsymbol{\alpha }}_m}\left( {{\theta _{p,m}},{\varphi _{p,m}}} \right) \otimes {{\boldsymbol{\tau }}_m}\left( {{\tau _{p,m}}} \right)$ is an angle-delay signature of the $p$-th path. The large scale channel gain is $\beta_{p,m}$. The 2D steering vector ${{\boldsymbol{\alpha }}_m}\left( {{\theta _{p,m}},{\varphi _{p,m}}} \right)$ is a Kronecker product of the horizontal steering and vertical steering vector \cite{2013FD}. ${\theta _{p,m}}$, ${\varphi _{p,m}}$ are the zenith angle and azimuth angle, respectively. The total number of paths is $P$. The Doppler frequency is ${\omega _{p,m}} = 2\pi {{v\cos {\phi _{p,m}}{f_0}} \mathord{\left/
 {\vphantom {{v\cos {\phi _{p,m}}{f_0}} c}} \right.
 \kern-\nulldelimiterspace} c}$, where $c$ is the speed of light and ${\phi _{p,m}}$ is the angle between the path $p$ and the speed vector of the UE $k$. The delay vector ${{\boldsymbol{\tau }}_m}\left( {{\tau _{p,m}}} \right)$ characterizes the delay information ${{\tau _{p,m}}}$ at all subcarriers
\begin{equation}
    {{\boldsymbol{\tau }}_m}\left( {{\tau _{p,m}}} \right) = {e^{2\pi j{\tau _{p,m}}{f_0}}}\left[ {\begin{array}{*{20}{c}}
1& \cdots &{{e^{j2\pi {\tau _{p,m}}\left( {{N_f} - 1} \right){\Delta_f}}}}
\end{array}} \right].
\nonumber
\end{equation}

According to the virtual channel representation idea \cite{dft2002}, the channel \eqref{wideband channel} can be projected into the angle-delay domain.
Define a transform matrix as the Kronecker product of a discrete Fourier transform (DFT) matrix ${{\mathbf{W}}_1}\left( {{N_f}} \right) \in {{\mathbb{C}}^{{N_t} \times {N_t}}}$ and a DFT matrix ${{\mathbf{W}}_2}\left( {{N_f}} \right) \in {{\mathbb{C}}^{{N_f} \times {N_f}}}$
\begin{equation}\label{projection matrix}
    {\bf{W}}\left( {{N_f},{N_t}} \right){\rm{ = }}{{\bf{W}}_1}{\left( {{N_f}} \right)^H} \otimes {{\bf{W}}_2}\left( {{N_t}} \right).
\end{equation}
Then the downlink channel in angle-delay domain ${{\bf{g}}_m}\left( t \right)$ is 
\begin{equation}\label{projected wideband channel}
    {{\bf{g}}_m}\left( t \right) = {\bf{W}}{\left( {{N_t},{N_f}} \right)^H}{{\bf{h}}_m}\left( t \right).
\end{equation}

As the matrix ${\bf{W}}{\left( {{N_t},{N_f}} \right)}$ is an unitary matrix, the wideband channel \eqref{wideband channel} can be equivalently calculated by 
\begin{equation}\label{wideband band in angle-delay path form}
    {{\bf{h}}_m}\left( t \right) = {\bf{W}}\left( {{N_t},{N_f}} \right){{\bf{g}}_m}\left( t \right) = \sum\limits_{n = 1}^{{N_t}{N_f}} {{\eta _{m,n}}\left( t \right){{\bf{w}}_n}},  
\end{equation}
where ${{\eta _{m,n}}\left( t \right)}$ is the $n$-th element of the vector ${{\bf{g}}_m}\left( t \right)$. It represents the amplitude of the corresponding angle-delay vector ${\bf{w}}_n$, which is the $n$-th column vector of ${\bf{W}}{\left( {{N_t},{N_f}} \right)}$. 
\subsection{Periodic eigenvector acquisition}
In our paper, the EZF precoding scheme is enabled by the SVD of the wideband channel matrix. To be more specific, the downlink precoder is selected from the eigenvectors of ${\bf{H}}\left( t \right)$. Due to the limited SVD computation capability of most practical communication systems, the BS often updates the precoders periodically with a cycle length $T_{\rm{svd}}$, which is longer than the duration of a subframe. 

We now discuss the eigenvector acquisition for precoding at the BS. First, the SVD of ${\bf{H}}\left( t \right)$ is given by
\begin{equation}\label{SVD equation}
    {{\bf{H}}}\left( t \right) = {\bf{U}}\left( t \right){\bf{\Sigma }}\left( t \right){\bf{V}}{\left( t \right)^H} = {{\bf{U}}_M}\left( t \right){{\bf{\Sigma }}_M}\left( t \right){{\bf{V}}_M}{\left( t \right)^H},
\end{equation}
where ${\bf{U}}\left( t \right)\in \mathbb{C}^{N_fN_t \times N_fN_t}$ is the left-singular matrix and ${\bf{V}}\left( t \right)\in \mathbb{C}^{M \times M}$ is the right-singular matrix. ${\bf{\Sigma }}\left( t \right)\in \mathbb{C}^{N_fN_t \times M}$ is the singular value matrix of which elements are the singular values sorted in a descending order. The subscript $M$ means $M$-truncated in columns. And the $M$-truncated eigenmatrix ${{\bf{U}}_M}\left( t \right)\in \mathbb{C}^{N_fN_t \times M}$ is defined by
\begin{equation}\label{eigen vector difinition}
   {{\bf{U}}_M}\left( t \right) = \left[ {\begin{array}{*{20}{c}}
{{{\bf{u}}_1}\left( t \right)}&{{{\bf{u}}_2}\left( t \right)}& \cdots &{{{\bf{u}}_M}\left( t \right)}
\end{array}} \right].
\end{equation}

Note that the eigenmatrix ${{\bf{U}}_M}\left( t \right)$ suffers from a random phase ambiguity due to the SVD nature. Therefore, a phase-calibration progress of the eigenvectors \eqref{eigen vector difinition} should be introduced.  
During a period of time ${t} \in \left[ {{t_{{\rm{in}}}},{t_{{\rm{ed}}}}} \right]$, the eigenvectors ${{\bf{u}}_m}\left( t \right)$ are phase-calibrated by
   \begin{equation}\label{phase calibration of SVD}
        \left[ {\begin{array}{*{20}{c}}
{{{\overline {\bf{u}} }_m}\left( {{t_{{\rm{in}}}}} \right)}\\
{{{\overline {\bf{u}} }_m}\left( {{t_{{\rm{in}}}} + {\Delta _t}} \right)}\\
 \cdots \\
{{{\overline {\bf{u}} }_m}\left( {{t_{{\rm{ed}}}}} \right)}
\end{array}} \right] = \left[ {\begin{array}{*{20}{c}}
{{{\bf{u}}_m}\left( {{t_{{\rm{in}}}}} \right)}\\
{{\Delta _m}\left( {{t_{{\rm{in}}}} + {\Delta _t}} \right){{\bf{u}}_m}\left( {{t_{{\rm{in}}}} + {\Delta _t}} \right)}\\
 \cdots \\
{{\Delta _m}\left( {{t_{{\rm{ed}}}}} \right){{\bf{u}}_m}\left( {{t_{{\rm{ed}}}}} \right)}
\end{array}} \right],
        \nonumber
    \end{equation} 
where the phase calibration is calculated by ${\Delta _m}\left( {{t}} \right) = \frac{{{{\bf{u}}_m}{{\left( {{t}} \right)}^H}{{\bf{u}}_m}\left( {{t_{{\rm{in}}}}} \right)}}{{\left| {{{\bf{u}}_m}{{\left( {{t}} \right)}^H}{{\bf{u}}_m}\left( {{t_{{\rm{in}}}}} \right)} \right|}}$. The initial subframe and ending subframe are $t_{\rm{in}}$ and $t_{\rm{ed}}$, respectively. The current subframe is $t$. 

Obviously, the phase-calibrated eigenvector $\overline {\bf{u}}_m \left( t \right)$ is also the eigenvector of ${{\bf {H}}}\left( t \right)$. Learned from basic algebra knowledge, the phase-calibration does not affect the SE.

\begin{figure}[ht]
\centering
\includegraphics[width=3.2in]{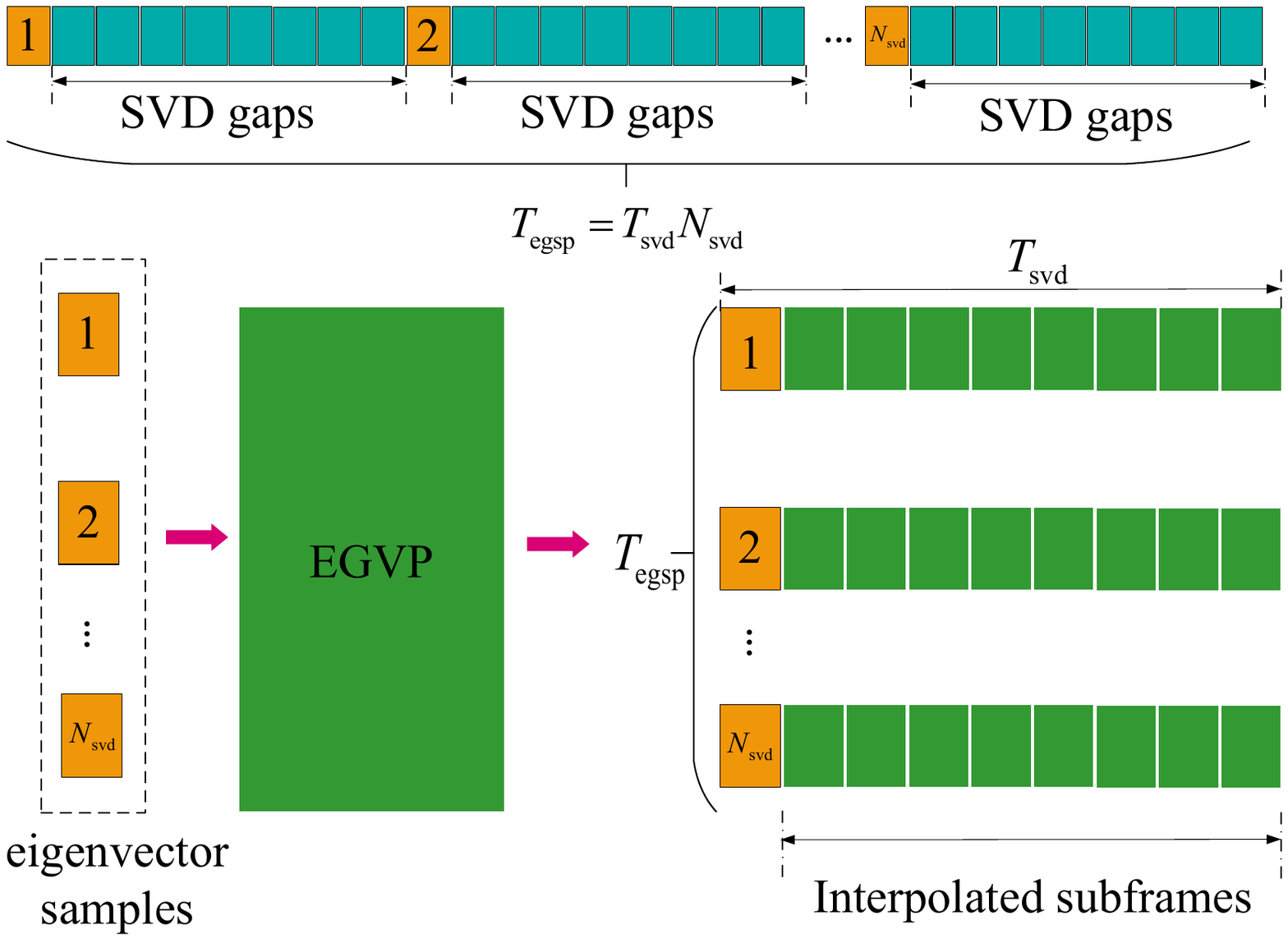}
\caption{Illustration of the precoding matrix prediction in the EGVP scheme within a period of $T_{\rm{egsp}}$ subframes.}
\vspace{-0.3cm}
\label{fig_framesturcture}
\end{figure}
As the result of the periodic SVD at the BS, eigenvectors are sampled every $T_{\rm{svd}}$ subframes and the rest of them need to be interpolated. In this case, our EGVP scheme is devised to interpolate the precoders when the SVD operation is not available. We refers these subframes as SVD gaps in the rest of the paper. The basic idea of EGVP lies in utilizing a small amount $N_{\rm{svd}}$ of eigenvector samples to interpolate the rest of eigenvectors within a period of time $T_{\rm{egsp}}=N_sT_{\rm{svd}}$ which is referred to as the EGVP interval in the following. The illustration of our EGVP within one $T_{\rm{egsp}}$ is shown in Fig. \ref{fig_framesturcture} and the details will be discussed in the next section.

\section{EGVP with perfect CSI}\label{sec3}
In this section, the detailed algorithm of the EGVP method is derived considering the ideal case where the downlink CSI is known. The realistic case where the CSI suffers from a non-negligible delay will be discussed in the next section. Given a perfect CSI, the precoder prediction problem is transformed into a channel weight interpolation problem. First, the calibrated eigenmatrix  $\overline {\bf{U}}_M \left( t \right)$ is written as a linear combination of the channels and channel weights. It is then shown that the channel weight can be approximated by a exponential model, which enables a low-complexity channel weight interpolation. In the end, the precoders are predicted by the interpolated channel weights and the CSI. 
\subsection{Linear eigenvector decomposition}
The following proposition is first derived. 
\begin{proposition}\label{Proposition eigen decomposition}
   The eigenvectors $\overline {\bf{u}}_m \left( t \right)$ can be decomposed into a linear combination of the channels ${\bf{h}}_j\left( t \right)$ and channel weights $a_{m,j}\left(t\right)$
     \begin{equation}\label{eigenvector decomposition}
         \overline {\bf{u}}_m \left( t \right) = \sum\limits_{j = 1}^M {{a_{m,j}}\left( t \right){{\bf{h}}_j}} \left( t \right).
      \end{equation}
\end{proposition}  
\begin{proof}
\quad \emph{Proof:} Please refer to Appendix \ref{Appendix Proposition eigen decompose}.
\end{proof}
Based on \eqref{weight defination} in Appendix \ref{Appendix Proposition eigen decompose}, the channel weight $a_{m,j}\left( t \right)$ actually characterizes the ratio of the inner product of ${\bf{h}}_j\left(t\right)$ and $\overline {\bf{u}}_m \left( t \right)$ to the $m$-th eigenvalue $\chi_m\left(t\right)$ 
\begin{equation}\label{weight physical meaning}
    {a_{m,j}}\left( t \right) = \frac{{\left\langle {\overline {\bf{u}}_m \left( t \right),{{\bf{h}}_j}\left( t \right)} \right\rangle }}{{{\chi _m}\left( t \right)}}.
\end{equation}
Utilizing Proposition \ref{Proposition eigen decomposition}, the eigenmatrix is equivalently written  
\begin{equation}\label{eigen decompostion calculation}
    \overline {\bf{U}}_M \left( t \right) = \left[ {\begin{array}{*{20}{c}}
{{{\bf{h}}_1}\left( {t} \right)}&{{{\bf{h}}_2}\left( {t} \right)}& \cdots &{{{\bf{h}}_M}\left( {t} \right)}
\end{array}} \right]{\bf{A}}\left( t \right),
\end{equation}
where the channel weight matrix ${\bf{A}}\left( t \right)$ is 
\begin{equation}\label{weight matrix definition}
    {\bf{A}}\left( t \right) = \left[ {\begin{array}{*{20}{c}}
{{a_{1,1}}\left( t \right)}&{{a_{2,1}}\left( t \right)}& \cdots &{{a_{M,1}}\left( t \right)}\\
{{a_{1,2}}\left( t \right)}&{{a_{2,2}}\left( t \right)}& \cdots &{{a_{M,2}}\left( t \right)}\\
 \cdots & \cdots & \cdots & \cdots \\
{{a_{1,M}}\left( t \right)}&{{a_{2,M}}\left( t \right)}& \cdots &{{a_{M,M}}\left( t \right)}
\end{array}} \right].
\end{equation}

Therefore, the eigenvector interpolation problem is equivalent to the channel weight interpolation problem.
\subsection{Channel weight estimation model}
The interpolation of the channel weights is difficult to obtain due to the non-linearity and non-singularity nature of SVD. However, it might be asymptotically approximated by a complex exponential model. 
Unless otherwise specified, the asymptotic condition refers to the case where the bandwidth $N_f$ and the number of base station antennas $N_t$ approaches infinity. We will propose a theorem to deal with the channel weight estimation problem. To begin with, we introduce the following lemma to prove that the inner product of two downlink channels ${s_{i,j}}\left( t \right) = \left\langle {{{\bf{h}}_j}\left( t \right),{{\bf{h}}_i}\left( t \right)} \right\rangle$ can be approximated by a complex exponential model.

\begin{lemma}\label{Lemma of relevance model}
    The channel inner product ${s_{i,j}}\left( t \right)$ can be asymptotically approximated by a complex exponential model
    \begin{equation}\label{channel inner product DFT}
         \mathop {\lim }\limits_{{N_t}{N_f} \to \infty } {{s_{i,j}}\left( t \right) = \sum\limits_{x \in {{\cal{X}}_{i,j}}} {\rho _{i,j}^{\left( x \right)}{e^{j\omega _{i,j}^{\left( x \right)}t}}} },
    \end{equation}
    where ${\rho _{i,j}^{\left( x \right)}}$ is the complex amplitude and ${{e^{j\omega _{i,j}^{\left( x \right)}t}}}$ is the corresponding exponential. The set ${\cal{X}}_{i,j}$ contains the indices of non-zero angle-delay vectors and the size of it is ${N_{{\rm{no}},i,j}}$.
\end{lemma}
\begin{proof}
\quad \emph{Proof:} Please refer to Appendix \ref{Appendix for lemma of relevance model}.
\end{proof}
According to Appendix \ref{Appendix for lemma of relevance model}, the physical meaning of $N_{{\rm{no}},i,j}$ is the number of non-orthogonal multi-paths between the channel ${\bf{h}}_i\left( t \right)$ and ${\bf{h}}_j\left( t \right)$. The exponential ${{e^{j\omega _{i,j}^{\left( x \right)}t}}}$ reflects the difference of Doppler frequency shift between the non-orthogonal path $x$ of the channel ${\bf{h}}_i\left( t \right)$ and ${\bf{h}}_j\left( t \right)$. 
We introduce a technical assumption below on the asymptotic value of ${s_{i,j}}\left( t \right)$. 

\begin{assumption}\label{assumption for asymptotic analysis}
The asymptotic value of ${s_{i,j}}\left( t \right)$ satisfies
\begin{small}
 \begin{equation}\label{asympototic results channel one}
     \mathop {\lim }\limits_{{N_t}{N_f} \to \infty } \frac{{{s_{i,i}}\left( t \right)}}{{{N_t}{N_f}}} = \mathop {\lim }\limits_{{N_t}{N_f} \to \infty } \frac{{\left\| {{{\bf{h}}_i}\left( t \right)} \right\|_2^2}}{{{N_t}{N_f}}} = {\cal{A}}_i,
 \end{equation}
\begin{equation}\label{asympototic results channel two}
    \mathop {\lim }\limits_{{N_t}{N_f} \to \infty } \frac{{{s_{j,j}}\left( t \right)}}{{{N_t}{N_f}}} = \mathop {\lim }\limits_{{N_t}{N_f} \to \infty } \frac{{\left\| {{{\bf{h}}_j}\left( t \right)} \right\|_2^2}}{{{N_t}{N_f}}} = {\cal{B}}_j,
\end{equation}
\begin{equation}\label{asympototic results channel relevance}
    \mathop {\lim }\limits_{{N_t}{N_f} \to \infty } \frac{{ \left|{{s_{i,j}}\left( t \right)} \right|}}{{{N_t}{N_f}}} = \mathop {\lim }\limits_{{N_t}{N_f} \to \infty } \frac{{\left| {\left\langle {{{\bf{h}}_i}\left( t \right),{{\bf{h}}_j}\left( t \right)} \right\rangle } \right|}}{{{N_t}{N_f}}} = {\cal{C}}_{i,j} \ne 0.
\end{equation}
\end{small}
\end{assumption}

Based on the law of large numbers, the equations \eqref{asympototic results channel one} and \eqref{asympototic results channel two} can be reasonably assumed. The non-zero constants ${\cal{A}}_i$ and ${\cal{B}}_j$ mean that the normalized channel gain of ${\bf{h}}_i\left( t \right)$ and ${\bf{h}}_j\left( t \right)$ converge to constant values. The assumption \eqref{asympototic results channel relevance} is the result of the non-trivial correlation between the channel ${{\bf{h}}_i}\left( t \right)$ and ${{\bf{h}}_j}\left( t \right)$, which is caused by the multi-antennas at the UE side and the finite number of multi-paths of the channel. This assumption is valid under the channel model of 5G NR \cite{3gpp901}.

In the end, the following theorem is proposed to prove that the channel weight $a_{m,j}\left(t\right)$ can be asymptotically approximated by a complex exponential model.

\begin{theorem}\label{theorem weight prediction model}
     The channel weight $a_{m,j}\left(t\right)$ can be approximated by a linear combination of complex exponential functions under Assumption \ref{assumption for asymptotic analysis}. 
    \begin{equation}\label{weight prediction model}
        \mathop {\lim }\limits_{{N_t}{N_f} \to \infty } {a_{m,j}}\left( t \right){\rm{ }} = \sum\limits_{l = 1}^{{L_{i,j}}} {b_{m,j}^{\left( l \right)}{e^{j\omega _{m,j}^{\left( l \right)}t}}} ,
    \end{equation}
    where ${e^{j\omega _{m,j}^{\left( l \right)}t}}$ is the exponential and ${b_{m,j}^{\left( l \right)}}$ is the corresponding amplitude. The order $L_{i,j}$ is the number of exponentials.
\end{theorem}
\begin{proof}
\quad \emph{Proof:} Please refer to Appendix \ref{Appendix Theroem weight predicion model}.
\end{proof}

We should note that the complex exponential model of the channel weight in \eqref{weight prediction model} is not unique. The reason lies in the plain fact of non-unique eigenvector solution of \eqref{solving eigen equation}. Fortunately, Appendix \ref{Appendix Theroem weight predicion model} proves that this does not affect the conclusion of Theorem \ref{theorem weight prediction model}. Moreover, it gives a particular eigenvector solution ${\overline {\bf{a}} _m}\left( t \right)$ related to the channel inner product ${s_{i,j}}\left( t \right)$. In such case, the number of exponentials $L_{i,j}$ is associated with the number of non-orthogonal multi-paths $N_{{\rm{no}},i,j}$. The exponential
${e^{j\omega _{m,j}^{\left( l \right)}t}}$ characterizes the difference of Doppler shifts between the channel ${\bf{h}}_m\left( t \right)$ and ${\bf{h}}_j\left( t \right)$. 
The exact value of $L_{i,j}$ can be computed by a signal detecting technology \cite{chen1991detection}. However, a unified $L_{i,j}=L_{\rm{svd}}$ is also a practical option, which will be demonstrated in Sec. \ref{sec6}. 

Theorem \ref{theorem weight prediction model} is the basis of EGVP method. Then we discuss how to predict the precoders by the eigenvector interpolation.
\subsection{Precoder prediction}\label{eigenvector interpolation}  
In order to interpolate the channel weights within SVD gaps, we utilize $N_{\rm{svd}}$ eigenvector samples to obtain the channel weight estimation model in \eqref{weight prediction model}. The precoders are predicted by the interpolated eigenvectors and the CSI.

Denote the eigenvector samples within one EGVP interval $T_{\rm{egsp}}$ by 
\begin{small}
  \begin{equation}\label{eigen vector sample}
    \overline {\bf{u}}_m \left( {{t_{{\rm{inl}}}}} \right),\overline {\bf{u}}_m \left( {{t_{{\rm{inl}}}} + {T_{{\rm{svd}}}}} \right), \cdots, \overline {\bf{u}}_m \left( {{t_{{\rm{inl}}}} + \left( {{N_{\rm{svd}}} - 1} \right){T_{{\rm{svd}}}}} \right),
    \nonumber
\end{equation}     
\end{small}
where $t_{\rm{inl}}$ is the initial subframe. According to the physical meaning of the channel weight model \eqref{weight prediction model} in Theorem \ref{theorem weight prediction model}, the following corollary is easily derived to ensure the Nyquist sampling frequency.

\begin{corollary}\label{corollary}
    The SVD sample length is upper bounded by
    \begin{equation}\label{sampling constraint}
        {T_{{\rm{svd}}}} \le \frac{c}{{2{v_{\max }}{f_0}}},
    \end{equation}
    where the right hand side denotes the inverse of the maximum Doppler frequency of the channel.
\end{corollary}

This corollary points out an upper bound of the SVD cycle given a maximum UE velocity. Considering a TDD system with center frequency $f_0=$ 3.5 GHz  and subframe length $\Delta _t=$ 0.5 ms, if the maximum speed of UE is about 123 km/h, the SVD cycle should not exceed $T_{\rm{svd}}=1.5$ ms. In this case, the eigenvector samples update every 3 subframes. 

The channel weight estimation based on \eqref{weight prediction model} is a classic exponential parameter estimating problem. Many traditional algorithms, like Prony \cite{1982prony}, maximum likelihood \cite{1986ML}, Matrix Pencil \cite{1990MP}, are capable to solve this problem. Due to the low complexity and robustness to noise, we utilize Matrix Pencil method to estimate the complex exponential model \eqref{weight prediction model} under a prediction order constraint $N_{\rm{svd}}\ge2L_{\rm{svd}}$. The detailed Matrix Pencil algorithm is omitted in our paper and can be found in \cite{1990MP,2022ziao}. Based on the estimated complex exponential model \eqref{weight prediction model}, the interpolated channel weight ${{\hat a}_{m,j}}\left( {{t_p}} \right)$ at subframe $t_p$ is 
\begin{equation}\label{weight interpolation equation}
{{\hat a}_{m,j}}\left( {{t_p}} \right) = \sum\limits_{l = 1}^{{L_{{\rm{svd}}}}} {b_{m,j}^{\left( l \right)}{e^{j\frac{{\omega _{m,j}^{\left( l \right)}}}{{{T_{{\rm{svd}}}}}}{t_p}}}} , 
\end{equation}
where the subframe $t_{p}$ satisfies
\begin{equation}\label{interpolated interval}
 \bmod \left( {t_p - {t_{{\rm{in}}}},T_{\rm{svd}}} \right) \ne 0,t_p \in \left[ {{t_{{\rm{in}}}},{t_{{\rm{in}}}} + \left( {{N_{\rm{svd}}} - 1} \right){T_{{\rm{svd}}}}} \right].
 \nonumber
\end{equation} 
Combined with the CSI, the eigenvector is reconstructed by 
\begin{equation}\label{eigen vector reconstruction}
   \hat {\bf{u}}_m\left( {{t_p}} \right) = \sum\limits_{j = 1}^M {{{\hat a}_{m,j}}}\left( {{t_p}} \right){{\bf{h}}_j} \left( {{t_p}} \right). 
\end{equation}
The whole EGVP algorithm is summarized in Algorithm \ref{alg EGVP}. 
\begin{algorithm} [H]                    
\caption{EGVP algorithm with perfect CSI in one $T_{\rm{egsp}}$}          
\label{alg EGVP}                           
\begin{algorithmic} [1]
\normalsize           
\STATE Initialize $T_{\rm{svd}},L_{\rm{svd}},T_{\rm{egsp}}$, $t_{\rm{in}}$, $t_{\rm{ed}}$ and obtain ${\bf{H}}\left( t \right)$;
\STATE Obtain periodic eigenvector samples based on \eqref{SVD equation} and calibrate the phase as $\overline {\bf{u}}_m \left( t \right)$;
    \STATE Decompose the eigenvector $\overline {\bf{u}}_m \left( t \right)$ and obtain the  channel weights $a_{m,j}\left(t\right)$ with \eqref{eigenvector decomposition}; 
    \STATE Generate a prediction matrix with $a_{m,j}\left(t\right)$ 
    and estimate the model in \eqref{weight prediction model} using Matrix Pencil method; 
        \FOR{$t \in \left[ {{t_{{\rm{in}}}},{t_{{\rm{ed}}}}} \right] \cap \,\bmod \,\left( {t - {t_{{\rm{in}}}},{T_{{\rm{svd}}}}} \right) \ne 0$}
            \STATE Interpolate the channel weight utilizing \eqref{weight interpolation equation}
            and reconstruct the eigenvector with the interpolated channel weights and the CSI based on \eqref{eigen vector reconstruction};
            \STATE Update $t = t+1$;
        \ENDFOR
\STATE Return the eigenvectors $\hat {\bf{u}}_m( t )$ and obtain the precoders.
\end{algorithmic}
\end{algorithm}

In this section, we demonstrate the specific procedures of EGVP method where the perfect downlink CSI is available. However, in fast time-varying channel, the non-neglectable CSI delay problem declines the system performance. The delayed CSI will be dealt with in EGVG-FMPP method in the next section.
\section{EGVP-FMPP with delayed CSI}\label{sec4}
In Sec. \secref{sec3}, the precoders are predicted where perfect and timely CSI is assumed. The fast time-varying channel in practical systems makes it highly challenging to obtain the timely CSI for the BS. To be more specific, the delay of CSI includes three parts: the signal transmission delay ${T_{{\text{trs}}}^d}$, the SVD delay ${T_{{\text{svd}}}^d}$ and the interpolation delay of EGVP ${T_{{\text{int}}}^d}$. The total CSI delay is ${T^d} = T_{{\text{trs}}}^d + T_{{\text{int}}}^d + T_{{\text{svd}}}^d$. The signal transmission delay is associated with the Doppler effect. Moreover, the SVD delay cannot be overlooked due to the limited SVD capability of the base station.

In this section, we consider the more realistic setting with CSI delays. A channel prediction method, i.e., FMPP method is introduced in EGVP method to predict the precoders given the delayed CSI. It shows distinct complexity advantage over the traditional MP method \cite{1990MP}. The core idea of EGVP-FMPP lies in the CSI prediction which is based on the channel reciprocity in TDD and uplink channel parameter prediction before applying EGVP procedure. The predicted channel is
\begin{equation}\label{channel prediction model}
    {\hat {\bf{h}}_m}\left( {{t+T^d}} \right)= \sum\limits_{n = 1}^{{N_t}{N_f}} {{\eta _{m,n}}\left( {{t+T^d}} \right){{\bf{w}}_n}},
\end{equation}
where the complex amplitude ${{\eta _{m,n}}\left( t \right)}$ can be approximated by a summation of complex exponential functions \cite{2022ziao}
\begin{equation}\label{model of eta}
    {\eta _{m,n}}\left( t \right) = \sum\limits_{l = 1}^{{L_{{\rm{ce}}}}} {{\lambda _{m,n,l}}{e^{j{\omega _{m,n,l}}t}}}, 
\end{equation}
where ${L_{{\rm{ce}}}}$ is the number of Doppler shifts ${{e^{j{\omega _{m,n,l}}t}}}$. The complex amplitudes ${{\lambda _{m,n,l}}}$ denote the channel gains and random phases.
Considering the consistency of deployment, we still apply the idea of Matrix Pencil method here to acquire the exponential model \eqref{model of eta}. However, the proposed FMPP method shows a complexity advantage compared to the original Matrix Pencil method. The details of FMPP method are elaborated below.

Within a sampling interval $t \in \left[ {{t_{\rm{1}}},{t_{N_{{\rm{ce}}}}}} \right]$, the total ${N_{{\rm{ce}}}}$ history complex amplitudes ${\eta _{m,n}}\left( t \right)$ are utilized to construct two Hankel prediction matrices. The prediction order ${L_{{\rm{ce}}}}$ satisfies ${N_{{\rm{ce}}}} \ge 2 {L_{{\rm{ce}}}}$. For representation simplicity, the subscripts $\left({m,n}\right)$ are dropped below. One Hankel matrix is ${{\boldsymbol{\eta }}_1} \in {\mathbb{C}^{\left( {{N_{{\rm{ce}}}} - {L_{{\rm{ce}}}}} \right) \times {L_{{\rm{ce}}}}}}$ 
 \begin{equation}\label{channel prediction matrix one}
{{\boldsymbol{\eta }}_1}{\rm{ = }}\left[ {\begin{array}{*{20}{c}}
{\eta \left( {{t_{{L_{{\rm{ce}}}} + 1}}} \right)}&{\eta \left( {{t_{{L_{{\rm{ce}}}}}}} \right)}& \cdots &{\eta \left( {{t_2}} \right)}\\
{\eta \left( {{t_{{L_{{\rm{ce}}}} + 2}}} \right)}&{\eta \left( {{t_{{L_{{\rm{ce}}}} + 1}}} \right)}& \cdots &{\eta \left( {{t_3}} \right)}\\
 \vdots & \vdots & \ddots & \vdots \\
{\eta \left( {{t_{{N_{{\rm{ce}}}}}}} \right)}&{\eta \left( {{t_{{N_{{\rm{ce}}}} - 1}}} \right)}& \cdots &{\eta \left( {{t_{{N_{{\rm{ce}}}} - {L_{{\rm{ce}}}} + 1}}} \right)}
\end{array}} \right].
\nonumber
\end{equation}   

Similarly, the other is 
${{\boldsymbol{\eta }}_0} \in {\mathbb{C}^{\left( {{N_{{\rm{ce}}}} - {L_{{\rm{ce}}}}} \right) \times {L_{{\rm{ce}}}}}}$ 
 \begin{equation}\label{channel prediction matrix other}
{{\boldsymbol{\eta }}_0}{\rm{ = }}\left[ {\begin{array}{*{20}{c}}
{\eta \left( {{t_{{L_{{\rm{ce}}}}}}} \right)}&{\eta \left( {{t_{{L_{{\rm{ce}}}} - 1}}} \right)}& \cdots &{\eta \left( {{t_1}} \right)}\\
{\eta \left( {{t_{{L_{{\rm{ce}}}} + 1}}} \right)}&{\eta \left( {{t_{{L_{{\rm{ce}}}}}}} \right)}& \cdots &{\eta \left( {{t_2}} \right)}\\
 \vdots & \vdots & \ddots & \vdots \\
{\eta \left( {{t_{{N_{{\rm{ce}}}} - 1}}} \right)}&{\eta \left( {{t_{{N_{{\rm{ce}}}} - 2}}} \right)}& \cdots &{\eta \left( {{t_{{N_{{\rm{ce}}}} - {L_{{\rm{ce}}}}}}} \right)}
\end{array}} \right].
\nonumber
\end{equation}   

The following relationship holds according to \cite{1990MP}
   \begin{equation}\label{mp prediction solve first}
{{\boldsymbol{\eta }}_0} = {{\bf{Z}}_1}{\bf{\Lambda }}{{\bf{Z}}_2},{{\boldsymbol{\eta }}_1} = {{\bf{Z}}_1}{\bf{\Lambda }}{{\bf{Z}}_0}{{\bf{Z}}_2},
\end{equation} 
where the matrices ${{{\bf{Z}}_0}}\in {\mathbb{C}^{{L_{{\rm{ce}}}} \times {L_{{\rm{ce}}}}}}$, ${{{\bf{Z}}_1}}\in {\mathbb{C}^{\left( {{N_{{\rm{ce}}}} - {L_{{\rm{ce}}}}} \right) \times {L_{{\rm{ce}}}}}}$ and ${{{\bf{Z}}_2}}\in {\mathbb{C}^{{L_{{\rm{ce}}}} \times  { {L_{{\rm{ce}}}}} }}$ are obtained by the exponentials ${{e^{j{\omega _{m,n,l}}t}}}$ according to the definitions in \cite{1990MP}. The diagonal amplitude matrix is ${\bf{\Lambda }}$ and the diagonal elements of it are ${{\lambda _{m,n,l}}}$. Based on the definition,  ${{\bf{Z}}_1}$ is of full column rank  and ${{\bf{Z}}_2}$ is of full row rank. Hence, Eq. \eqref{mp prediction solve first} is equivalently written as
\begin{equation}\label{mp prediction solution recruisive}
{\bf{\Lambda }} = {{\bf{Z}}_1}^\dag {{\boldsymbol{\eta }}_0}{{\bf{Z}}_2}^\dag ,{{\boldsymbol{\eta }}_1} = {{\bf{Z}}_1}{\bf{\Lambda }}{{\bf{Z}}_0}{{\bf{Z}}_2} = {{\boldsymbol{\eta }}_0}{{\bf{Z}}_2}^\dag {{\bf{Z}}_0}{{\bf{Z}}_2}.
\end{equation}
Therefore, the following relationship can be derived
\begin{align}\label{prediction model of channel}
{{\boldsymbol{\eta }}_{{T^d} + 1}} &= {{\boldsymbol{\eta }}_{{T^d} - 1}}{{\boldsymbol{Z}}_2}^\dag {{\bf{Z}}_0}{{\bf{Z}}_2} = {{\boldsymbol{\eta }}_{{T^d} - 2}}{{\bf{Z}}_2}^\dag {{\bf{Z}}_0}^2{{\bf{Z}}_2} =  \cdots \\
 \nonumber
 &= {{\boldsymbol{\eta }}_0}{{\bf{Z}}_2}^\dag {\left( {{{\bf{Z}}_0}} \right)^{{T^d}}}{{\bf{Z}}_2}.
\end{align}
As a result, the element of the first column and last row of ${{\boldsymbol{\eta }}_{{T^d} + 1}}$ is the predicted amplitude by $T^d$ subframes, i.e., ${\eta _{m,n}}\left( {{t_d}} \right)$, where ${t_d} = {t_1} + {N_{{\rm{ce}}}} + {T^d}$. Hence, the predicted channel ${\hat {\bf{h}}_m}\left( {{t_d}} \right)$ is obtained by \eqref{channel prediction model}. Similar to EGVP method in \secref{sec3}, substitute the channels with ${\hat {\bf{h}}_m}\left( {{t_d}} \right)$ in \eqref{eigen vector reconstruction} and the precoders can be predicted. In this case, the EGVP method with predicted CSI is called EGVP-FMPP in our paper. 

The complexity advantage of the FMPP over the traditional MP method \cite{1990MP} is distinct. The complexity of the MP is ${\mathcal{O}}\left( {{N_t}{N_f}\left( {{N_{{\text{ce}}}}^2 + {L_{{\text{ce}}}}^2 + {L_{{\text{ce}}}}\log {T^d}} \right)} \right)$ whereas the complexity of the FMPP is ${\mathcal{O}}\left( {{N_t}{N_f}\left( {{L_{{\text{ce}}}}^2\left( {{N_{{\text{ce}}}} - {L_{{\text{ce}}}}} \right) + {L_{{\text{ce}}}}\log {T^d}} \right)} \right)$. Since ${N_{{\text{ce}}}} \geqslant 2{L_{{\text{ce}}}}$, the complexity reduction from the MP to the FMPP is ${\mathcal{O}}\left( {{N_t}{N_f}\left( {{N_{{\text{ce}}}}^2 + {L_{{\text{ce}}}}^2 + {L_{{\text{ce}}}}^3 - {L_{{\text{ce}}}}^2{N_{{\text{ce}}}}} \right)} \right) \geqslant {\mathcal{O}}\left( {{N_t}{N_f}{L_{{\text{ce}}}}^2\left( {5 - {L_{{\text{ce}}}}} \right)} \right)$. Therefore, given $L_{\rm{ce}} \le 5$, the FMPP method shows a distinct complexity advantage. Moreover, this condition is often valid due to the channel sparsity in massive MIMO systems.

The proposed FMPP method is also different from our previous channel prediction method, joint angel-delay-Doppler (JADD) in FDD \cite{2022ziao}. To be more specific, the downlink channel prediction in JADD scheme depends on the feedback amplitudes and the estimated exponetials. On the contrary, FMPP method only relies on the estimated exponetials predicted by the uplink channel samples, which means that there is no need to compute the amplitudes.

In the next section, the theoretical performances of our two EGVP schemes are demonstrated.



\section{Performance analysis}\label{sec5}
In this section, we first analyze how many samples are sufficient to achieve asymptotically error-free channel prediction and eigenvector prediction. Then a complexity analysis of our EGVP method is given to demonstrate its advantage over the state of the art (SoTA) schemes. 
\subsection{Prediction error performance}
First, the asymptotic channel prediction performance is derived in the following theorem. 
\begin{theorem}\label{theorem for asmptotic channel estimation}
     When ${N_t},{N_f} \to \infty$ and the number of channel samples ${N_{{\rm{ce}}}} = 2$, the channel prediction ${\hat {\bf{h}}_m}\left( t \right)$ obtained by FMPP is an asymptotically error-free estimation.  
     \begin{equation}\label{asmptotic channel prediction error conclusion}
            \begin{array}{l}
            \mathop {\lim }\limits_{{N_t},{N_f} \to \infty } \mathbb{E}\left\{ {\frac{{\left\| {{{\bf{h}}_m}\left( t \right) - {{\hat {\bf{h}}}_m}\left( t \right)} \right\|_2^2}}{{\left\| {{{\bf{h}}_m}\left( t \right)} \right\|_2^2}}} \right\}
             = 0.
            \end{array}
    \end{equation}
\end{theorem}
\begin{proof}
    \quad \emph{Proof:} Please refer to Appendix \ref{Appendix theorem for asmptotic channel estimation}.
\end{proof}

Theorem \ref{theorem for asmptotic channel estimation} shows that two history channel samples are enough to asymptotically achieve an error-free channel prediction in FMPP method. Furthermore, the PE is irrelevant with the channel delay $T^d$, which means that the effect of delayed CSI can be compensated. 

We should note that this conclusion is based on the assumption that the angle-delay-Doppler structure of the channel barely varies within $T^d$. This assumption usually holds in moderate mobility \cite{2020yinMobility,2022ziao}, e.g., $T^d=5$ ms, $v=60$ km/h.  

Then, we aim to analyze the eigenvector PE of EGVP which is measured by the normalized mean square error (NMSE) between the eigenvectors ${\overline {\bf{u}}_m \left( t \right)}$ and the interpolated eigenvectors ${{{\hat {\bf{u}}}_m}\left( t \right)}$. 
\begin{theorem}\label{theorem for asmptotic analysis of EGVP}
    When ${N_t},{N_f} \to \infty$ and the number of eigenvector samples $N_{\rm{svd}}\ge 2N_{\rm{no}}^{\rm{max}}$, EGVP method leads to an asymptotically error-free eigenvector prediction
    \begin{equation}\label{asmptotic EGVP prediction error}
          \mathop {\lim }\limits_{{N_t},{N_f} \to \infty } \mathbb{E}\left\{ {\left\| {\frac{{{{\overline {\bf{u}} }_m}\left( t \right) - {{\hat {\bf{u}}}_m}\left( t \right)}}{{{{\overline {\bf{u}} }_m}\left( t \right)}}} \right\|_2^2} \right\} = 0,
    \end{equation}
  where $N_{\rm{no}}^{\rm{max}}$ denotes the maximum number of non-orthogonal multi-paths among all downlink channels.
\end{theorem}
\begin{proof}
    \quad \emph{Proof:} Please refer to Appendix \ref{Appendix theorem for asmptotic eigen prediction}.
\end{proof}  

Theorem \ref{theorem for asmptotic analysis of EGVP} gives a minimum number of eigenvector samples to achieve an error-free EGVP. It is associated with the maximum number of non-orthogonal multi-paths among all downlink channels. In fact, it is exhausting to find the exact $N_{\rm{no}}^{\rm{max}}$ in practice due to the complicated channel environments. Even though it is possible to detect the size of $N_{\rm{no}}^{\rm{max}}$ \cite{2008detect}, a fixed sample quantity is more efficient. Furthermore, a small $N_{\rm{no}}^{\rm{max}}$ means more frequent eigenvector interpolation and is possible thanks to the channel sparsity of wideband massive MIMO \cite{2013YinJSAC}. Therefore, we tend to configure a small fixed $N_{\rm{no}}^{\rm{max}}$ and will testify the performance with numerical results in the next section. 

\subsection{Complexity analysis}
One of the major motivations of EGVP scheme is to reduce the complexity of the precoding matrix acquisition at the BS. To demonstrate the complexity advantage of EGVP scheme, four SoTA schemes are introduced for comparison, i.e., full-time SVD, periodic SVD, Wiener prediction and approximated gram matrix interpolation (AGMI) scheme. The downlink CSI is prior-known by the BS in the complexity analysis.

The full-time SVD scheme refers to the traditional scheme where the precoders are updated in every subframes like \cite{2020yinMobility,2022JSACParallel}. The complexity order of it is ${\mathcal{O}}\left( {{N_{{\text{svd}}}}{T_{{\text{svd}}}}M{{\left( {{N_f}{N_t}} \right)}^2}} \right)$.

The periodic SVD scheme means that a periodic precoder updating scheme is adopted at the BS, similar to the current precoding scheme in 5G NR \cite{3gpp214}. AGMI scheme can be utilized to reduce the complexity of precoding in the frequency domain \cite{2020gramesti}, but here, the eigenvectors are interpolated in the time domain
\begin{small}
   \begin{equation}\label{AGMI method}
{\hat {\mathbf{u}}_1}\left( {{t_p}} \right) = {\hat {\mathbf{u}}_1}\left( {{t_k}} \right)\left( {\frac{{{t_{k + 1}} - {t_p}}}{{{T_{{\text{svd}}}}}}} \right) + {\hat {\mathbf{u}}_1}\left( {{t_{k + 1}}} \right)\left( {1 - \frac{{{t_{k + 1}} - {t_p}}}{{{T_{{\text{svd}}}}}}} \right),
\end{equation} 
\end{small}
where the interpolated subframe ${t_p}$ satisfies ${t_k} < {t_p} < {t_{k + 1}}$. ${t_k},{t_{k + 1}}$ are the subframes of the $k$-th and $k+1$-th eigenvector samples, respectively. The complexity order of periodic SVD and AGMI scheme are the same as ${\mathcal{O}}\left( {{N_{{\text{svd}}}}M{{\left( {{N_f}{N_t}} \right)}^2}} \right)$.

Wiener prediction method is a classic interpolation method for discrete linear time-invariant system \cite{1997wierner,2015Wiener}. It first utilizes the prior-known auto-correction sequences ${{\bf{r}}_a}\left( \tau  \right)$ of the signals to calculate all $L_w$ order wiener filter coefficients through the Wiener-Hopf equation \cite{hayes1996statistical}
\begin{equation}\label{wiener filter coefficients}
    {\bf{w}} = {{\bf{R}}_a}{\left( \tau \right)^\dag }{{\bf{r}}_a}\left( {\tau + {T^d}} \right),
\end{equation}
where ${{\bf{R}}_a}\left( \tau \right)$ is the Hermitian Toeplitz matrix of ${{\bf{r}}_a}\left( \tau \right)$. The number of interpolated subframes is $T^d$. 

In this paper, the precoders are predicted with a linear combination of the $L_w=4$ wiener filter coefficients and recent $L_w$ eigenvector samples
\begin{equation}\label{wiener prediction equation}
    {{{u}}_m}\left( {t + {T^d}} \right) = \sum\limits_{l = 1}^{{L_w}} {w\left( l \right){{{u}}_m}\left( {t - l} \right)}, 
\end{equation}
where $w\left( l \right)$ is the element of ${\bf{w}}$ and ${{u_m}\left( t \right)},t \in \left[ {0,T - 1} \right]$ is the element of the eigenvector ${{{\bf{u}}_m}\left( t \right)}$. $T$ is the period of time. The auto-correction sequences are obtained by 
  \begin{equation}\label{difination of wiener autocorrelation}
    {{\bf{r}}_a}\left( \tau  \right) = {\mathbb{E}}\left\{ {{u_m}\left( t \right){u_m}{{\left( {t - \tau } \right)}^*}} \right\},\tau  \in \left[ { - T + 1,T - 1} \right].
    \nonumber
\end{equation}  
Therefore, the complexity order of Wiener prediction is ${\mathcal{O}}\left( {{N_{{\text{svd}}}}M{{\left( {{N_f}{N_t}} \right)}^2}} \right) + {\mathcal{O}}\left( {{N_{{\text{svd}}}}^2{T_{{\text{svd}}}}^2\left( {{N_f}{N_t}} \right)} \right)$. 

According to \secref{sec3}, the complexity order of EGVP is ${\mathcal{O}}\left( {{N_{{\text{svd}}}}M{{\left( {{N_f}{N_t}} \right)}^2}} \right) + {\mathcal{O}}\left( {{N_{{\text{svd}}}}{T_{{\text{svd}}}}{M^2}\left( {{N_f}{N_t}} \right)} \right)$. Obviously, the complexity order of the full-time SVD scheme is larger than the others by ${\mathcal{O}}\left( {\left( {{T_{{\text{svd}}}} - 1} \right){{\left( {{N_f}{N_t}} \right)}^2}} \right)$. The complexity order of EGVP increases ${\mathcal{O}}\left( {{N_{{\text{svd}}}}{T_{{\text{svd}}}}{M^2}{N_f}{N_t}} \right)$ compared to the periodic SVD and AGMI scheme. However, compared to the Wiener prediction scheme, EGVP scheme reduces the complexity order by ${\mathcal{O}}\left( {{N_{{\text{svd}}}}{T_{{\text{svd}}}}{N_f}{N_t}\left( {{N_{{\text{svd}}}}{T_{{\text{svd}}}} - M} \right)} \right)$, which is often positive under the practical configuration $M \leqslant 4,{N_{{\text{svd}}}}{T_{{\text{svd}}}} \geqslant 4$. 
\begin{figure}[!t]
\centering
\includegraphics[width=3.2in]{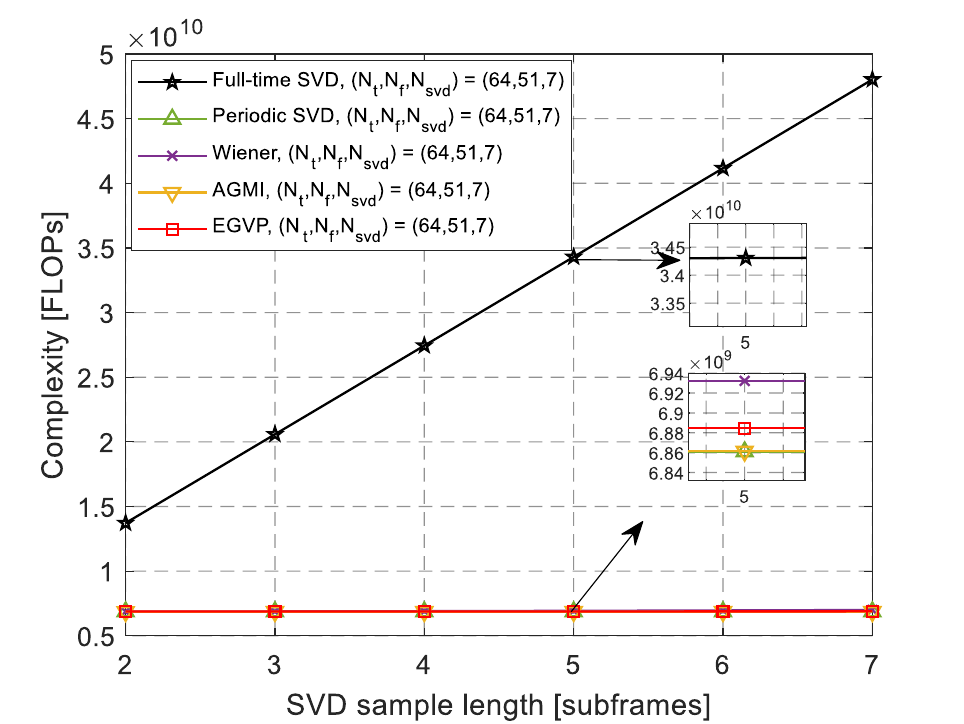}
\caption{Complexity comparison between EGVP and the other schemes in FLOPs under different configurations of the SVD cycle length and the number of eigenvector samples.}
\vspace{-0.3cm}
\label{fig_complexity}
\end{figure}

Furthermore, the complexity within one eigenvector sample cycle ${T_{{\text{egsp}}}} = {N_{{\text{svd}}}}{T_{{\text{svd}}}}$ is measured in floating point operations (FLOPs). Fig. \ref{fig_complexity} shows the complexity of different schemes in FLOPs under different configurations of the SVD cycle length and the number of eigenvector samples. Under a common configuration $\left( {{N_t},{N_f},{N_{{\text{svd}}}},{T_{{\text{svd}}}}} \right) = \left( {64,51,7,7} \right)$, the complexity of EGVP decreases 79.9\% compared to the full-time SVD scheme and $0.87\%$ compared to the Wiener prediction scheme. The complexity of EGVP only increases $0.29\%$ compared to the periodic SVD and AGMI scheme.  

As a result, the complexity advantage of EGVP scheme is explicit and the numerical results of the SE performance of EGVP are further shown in the next section.
\section{Simulation results}\label{sec6}
The numerical results of the proposed EGVP and EGVP-FMPP schemes are discussed in this section. A multi-path channel model widely used in industry, cluster delay line (CDL), is adopted in our simulations. Based on 3GPP specification \cite{3gpp901}, CDL-A channel model is generated with parameters in Table \ref{tab simulation config}. In this model, 23 clusters with 460 multi-paths are aggregated. The delay spread is 300 ns. The center frequency is 3.5 GHz. The bandwidth is 20 MHz and the subcarrier spacing (SCS) is 30 kHz. Hence, 51 resource blocks are configured and each subframe lasts 0.5 ms according to \cite{3gpp104}. The BS is equipped with a UPA configured as $\left( {{N_v},{N_h},{N_{pl}}} \right)= \left( {4,8,2} \right)$, where $N_v$ and $N_h$ are the number of antennas in one column and one row of the antenna array, respectively. The number of polarizations is denoted by $N_{pl}$. At the UE side, the receive antennas of all 8 UEs are configured as $\left( {{N_v},{N_h},{N_{pl}}} \right)= \left( {1,2,2} \right)$. We set the length of SVD cycle as 2.5 ms, i.e., the precoding matrix updates every 5 subframes. The EZF precoding is utilized at the BS and minimum mean mquare error-interference rejection combining (MMSE-IRC) is applied at the UEs. Unless other specified, we configure the prediction order of the eigenvector interpolation $L_{\rm{svd}}$ and channel prediction order $L_{\rm{ce}}$ both as 3. The number of the channel samples and the eigenvector samples are ${N_{{\text{ce}}}} = 2{L_{{\text{ce}}}} + 1$ and ${N_{{\text{svd}}}} = 2{L_{{\text{svd}}}} + 1$, respectively. The system performance is measured in SE or PE. The PE of the eigenvector is calculated by \eqref{asmptotic EGVP prediction error}. And the SE is computed 
 \begin{equation}\label{spectral efficiency}
{R_{{\rm{se}}}}{\rm{ = }}\mathbb{E}\left\{ {\sum\limits_{k = 1}^K {{\rm{log}}\left( {1 + \frac{{\left\| {{{\bf{h}}_k}\left( t,f \right){{\bf{g}}_k}\left( t,f \right)} \right\|_2^2}}{{\sigma _k^2 + \sum\limits_{j \ne k}^K {{{\left| {{{\bf{h}}_k}\left( t,f \right){{\bf{g}}_j}\left( t,f \right)} \right|}^2}} }}} \right)} } \right\},
\nonumber
\end{equation}   
where ${{{\bf{g}}_k}\left( {t,f} \right)}$ is the precoder of the $k$-th UE calculated by EZF and $\sigma _k^2$ is the power of Gaussian noise at the receiver side. The expectation is taken over subcarriers and subframes. The SNR at the UE varies from 0 dB to 30 dB.
\begin{table}[!t]
\centering \protect\protect\caption{System Parameters in Simulations}
\label{tab simulation config}
\begin{tabular}{|c|c|}
\hline
Channel model & CDL-A\tabularnewline
\hline
Bandwidth $B$ & 20 $\rm{MHz}$\tabularnewline
\hline
UL/DL carrier frequency $f_0$ & 3.5 $\rm{GHz}$\tabularnewline
\hline
Subcarrier spacing $\Delta_f$ & 30 $\rm{kHz}$\tabularnewline
\hline
Subframe duration ${\Delta _t}$& 0.5 $\rm{ms}$\tabularnewline
\hline
Resource block $N_f$ & 51 RB\tabularnewline
\hline
Delay spread $\tau_s$ & 300 $\rm{ns}$\tabularnewline
\hline
Number of paths $P$ & 460\tabularnewline
\hline
\tabincell{c}{Transmit antenna \\$\left( {{N_v},{N_h},{N_{pl}}} \right)$} & \tabincell{c}{$\left( {{N_v},{N_h},{N_{pl}}} \right)= \left( {4,8,2} \right)$, \\ the polarization directions are ${0^\circ },{90^\circ }$}\tabularnewline
\hline
\tabincell{c}{Receive antenna \\$\left( {{N_v},{N_h},{N_{pl}}} \right)$} & \tabincell{c}{$\left( {{N_v},{N_h},{N_{pl}}} \right)= \left( {1,2,2} \right)$, \\ the polarization directions are $\pm {45^\circ }$}\tabularnewline
\hline
Number of UEs $K$ & 8\tabularnewline
\hline
SVD cycle $T_{\rm{svd}}$ & 2.5 $\rm{ms}$\tabularnewline
\hline
CSI delay $T_{{\text{trs}}}^d$ & 2.5 $\rm{ms}$\tabularnewline
\hline
EGVP interpolation delay $T_{{\text{int}}}^d$ & 1.5 $\rm{ms}$\tabularnewline
\hline
SVD delay $T_{{\text{svd}}}^d$ & 2.5 $\rm{ms}$\tabularnewline
\hline
Channel estimation order $L_{\rm{ce}}$ & 3 \tabularnewline
\hline
EGVP order $L_{\rm{svd}}$ & 3 \tabularnewline
\hline
\end{tabular}
\end{table}

We consider two CSI acquisition scenarios including the timely CSI case and the delayed CSI case. The CSI delay in the benchmarks is ${T^d} = T_{{\text{trs}}}^d + T_{{\text{svd}}}^d = 5$ ms while ${T^d} = T_{{\text{trs}}}^d + T_{{\text{svd}}}^d + T_{{\text{int}}}^d = 6.5$ ms in EGVP-FMPP scheme considering the channel weight interpolation algorithm. The performance upper bound throughout our simulation refers to an ideal scenario where SVD operates at every subframe, denoted by ``Full-time SVD''. As described in \secref{sec5}, the periodic SVD, Wiener prediction and AGMI schemes are evaluated for performance comparison.

The proposed schemes and benchmarks are evaluated in various scenarios, including different UE speeds, different CSI delays, different BS antenna configurations, 
different SVD cycle lengths and noisy channel sampling cases. Our EGVP scheme is testified under the perfect CSI assumption and EGVP-FMPP scheme indicates that a channel prediction method FMPP is applied to combat the CSI delay $T^d$. 

In Fig. \ref{fig_simulation_low_speed}, the SE performances of EGVP and EGVP-FMPP are demonstrated when $v$ = 30 km/h. When perfect CSI is available, our EGVP scheme approaches the upper bound performance while reducing 79.9\% of complexity. When the SNR equals 30 dB, the EGVP increases SE by 5.5\% over the AGMI scheme and by 15.8\% SE over the periodic SVD scheme, however introducing only 0.29\% additional complexity. The EGVP scheme reduces 0.87\% complexity and enhances 15.5\% SE compared to Wiener scheme. 
However, when the CSI is delayed, the SE of all schemes declines. It is noteworthy that our EGVP-FMPP scheme exhibits less degradation than the benchmarks. The SE performance improvement of EGVP-FMPP scheme varies from 13.2\% to 32.7\% compared to the benchmarks. Therefore, we can conclude that the two EGVP schemes function well in a moderate-mobility system despite the CSI delay.
\begin{figure}[!t]
\centering
\includegraphics[width=3.2in]{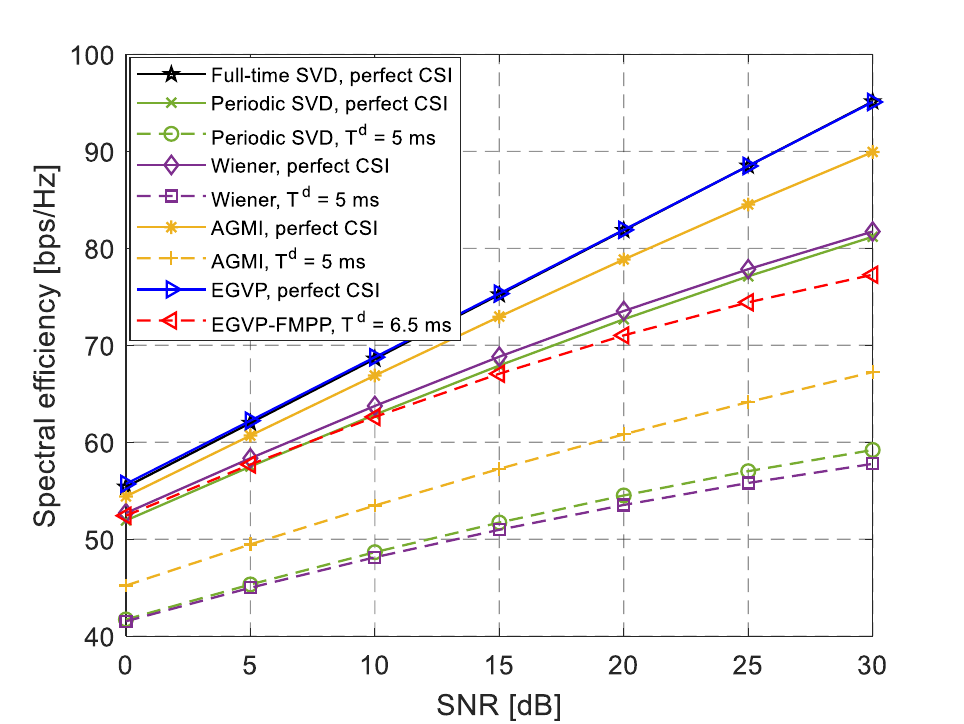}
\caption{SE performances with noise-free channel samples, $v=30$ km/h, $L_{\rm{ce}}=3$.}
\vspace{-0.3cm}
\label{fig_simulation_low_speed}
\end{figure}

Furthermore, the SE performance of the EGVP scheme is shown in Fig. \ref{fig_simulation_high_speed} when the UE speed is $v=60$ km/h. Based on the configuration in Table \ref{tab simulation config} and Corollary \ref{corollary}, the maximum UE velocity is $v_{\rm{max}}=61.8$ km/h. When  perfect CSI is assumed, the EGVP scheme approaches 97.8\% SE of the upper bound and increases 10.7\% over AGMI and 22.3\% SE over Wiener and the periodic SVD scheme. In case of delayed CSI, the orders of the channel estimator for the EGVP-FMPP scheme are configured to $L_{\rm{ce}}=3$ and $L_{\rm{ce}}=6$. Obviously, the larger $L_{\rm{ce}}$ results in better SE performance. It can be concluded that a higher order of channel estimator is more suitable in high mobility scenarios. In this scenario, the EGVP-FMPP scheme demonstrates a greater SE advantage over the benchmarks. To be more precise, when $L_{\rm{ce}}=6$, the EGVP-FMPP scheme enhances the SE performance from 20.6\% to 49.1\% in comparison to the benchmarks.  
\begin{figure}[!t]
\centering
\includegraphics[width=3.2in]{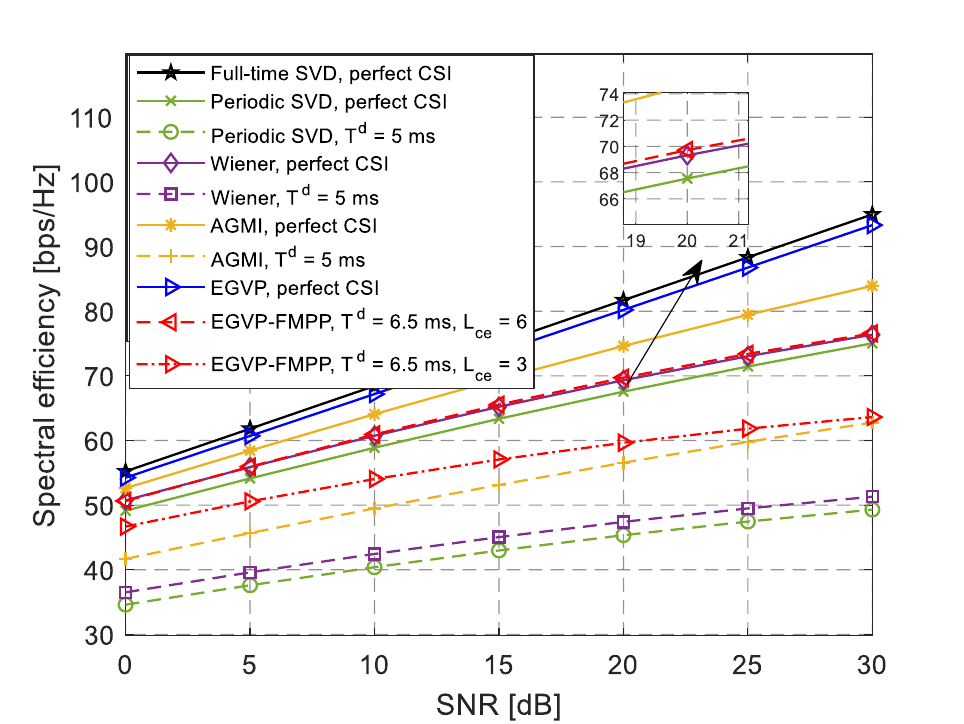}
\caption{SE performances with noise-free channel samples, $v=60$ km/h.}
\vspace{-0.3cm}
\label{fig_simulation_high_speed}
\end{figure}

Fig. \ref{fig_simulation_antenna} demonstrates the PE performances of the two EGVP scheme given different number of the BS antennas ${N_t} \in \left\{ {4,8,16,64,128,256} \right\}$ and different speeds of the UEs. In the EGVP-FMPP scheme, the PE decreases as $N_t$ increases. However, the PE of the EGVP scheme remains relatively stable with respect to $N_t$. Furthermore, the results demonstrate that both EGVP schemes are applicable for a relatively small $N_t$ which is more applicable in practical systems.  
\begin{figure}[!t]
\centering
\includegraphics[width=3.2in]{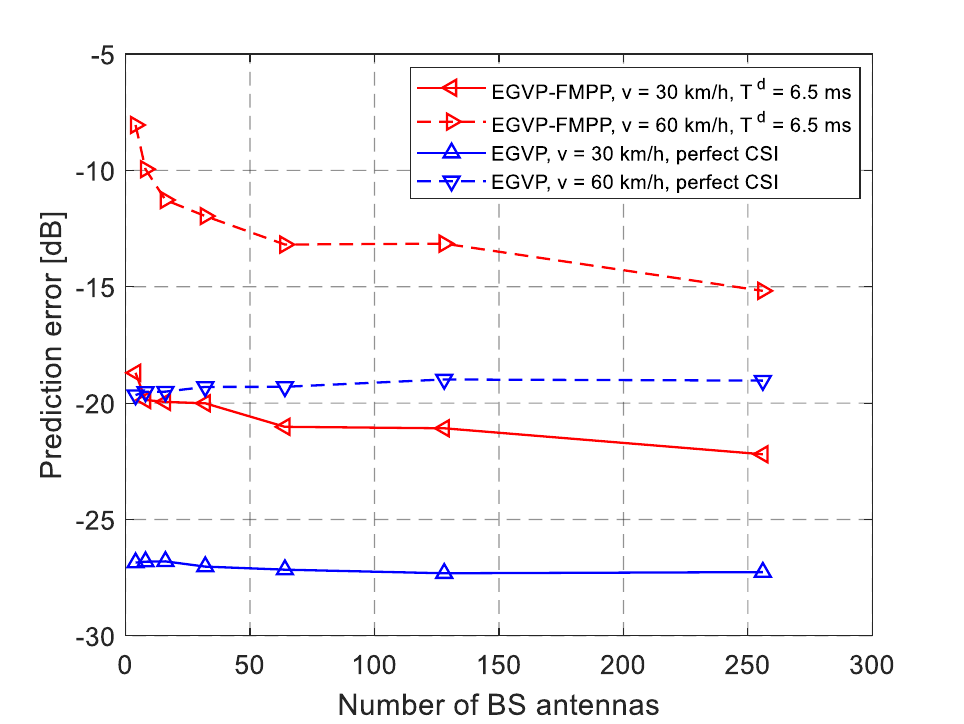}
\caption{PE performances vs the number of antennas with noise-free channel samples, $L_{\rm{ce}}=3$.}
\vspace{-0.3cm}
\label{fig_simulation_antenna}
\end{figure}

Then, we evaluate the PE performances given different SVD cycle lengths $T_{\rm{svd}}$ in Fig. \ref{fig_simulation_evd_cyc}. The PE performance of all schemes decreases slowly with $T_{\rm{svd}}$. Both EGVP schemes show distinct PE advantage over all benchmarks. No doubt a longer SVD cycle length leads to lower computation complexity yet larger PE. Therefore, the trade-off between the SVD cycle length and computation complexity can be tricky and the equation \eqref{sampling constraint} is helpful in this regard.
\begin{figure}[!t]
\centering
\includegraphics[width=3.2in]{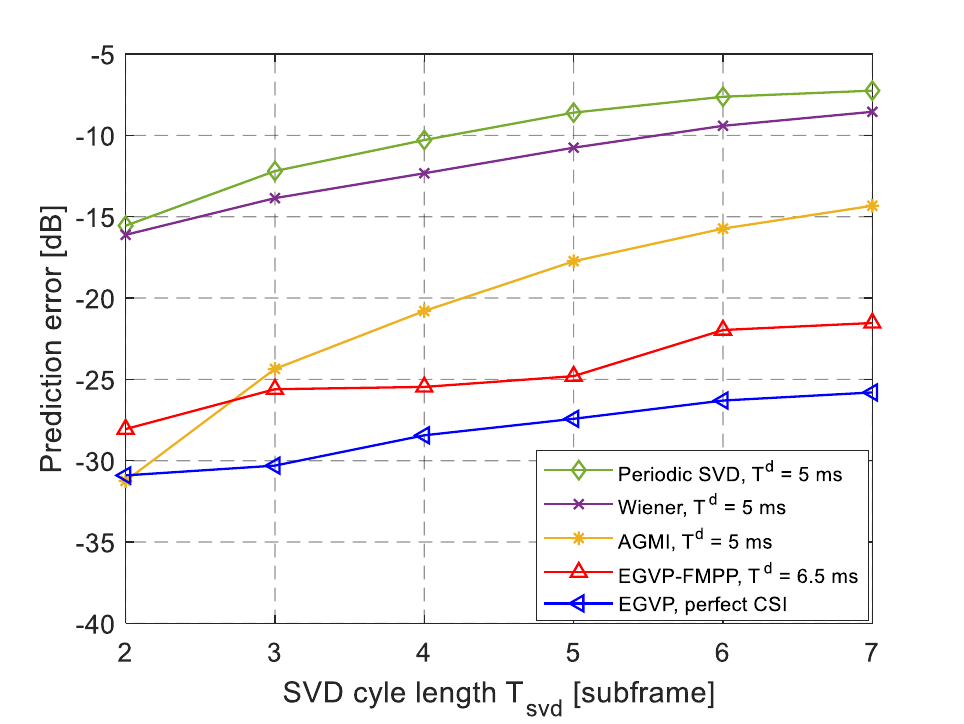}
\caption{PE performances of the eigenvectors vs different lengths of the SVD cycle with noise-free channel samples, $v=30$ km/h, $L_{\rm{ce}}=3$.}
\vspace{-0.3cm}
\label{fig_simulation_evd_cyc}
\end{figure}

\begin{figure}[!t]
\centering
\includegraphics[width=3.2in]{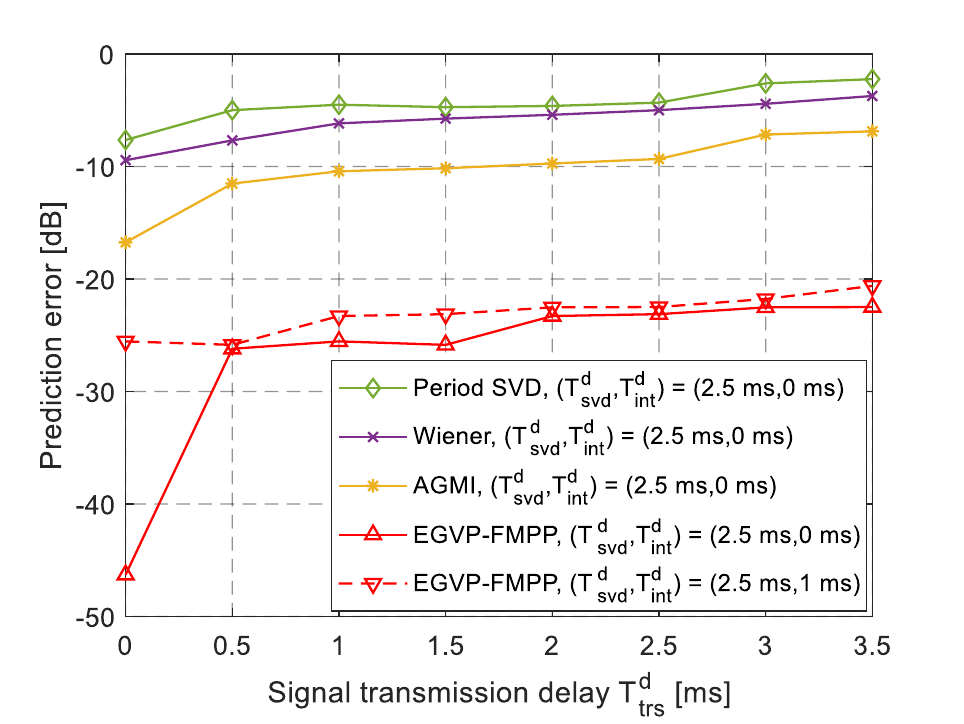}
\caption{PE performances of the eigenvectors vs different CSI delays with noise-free channel samples, $v=30$ km/h.}
\vspace{-0.3cm}
\label{fig_simulation_csiDelay}
\end{figure}
Fig. \ref{fig_simulation_csiDelay} illustrates the PE performances of EGVP-FMPP scheme and the benchmarks under different CSI delays. The benchmarks consider only the transmission delay $T_{{\text{trs}}}^d$ and the SVD delay $T_{{\text{svd}}}^d$ whereas EGVP-FMPP considers the additional interpolation delay $T_{{\text{int}}}^d$. It can be observed that as the CSI delay increases, the PE performance of all schemes deteriorates. However, EGVP-FMPP demonstrates superior PE performance compared to the benchmarks.



\begin{figure}[!t]
\centering
\includegraphics[width=3.2in]{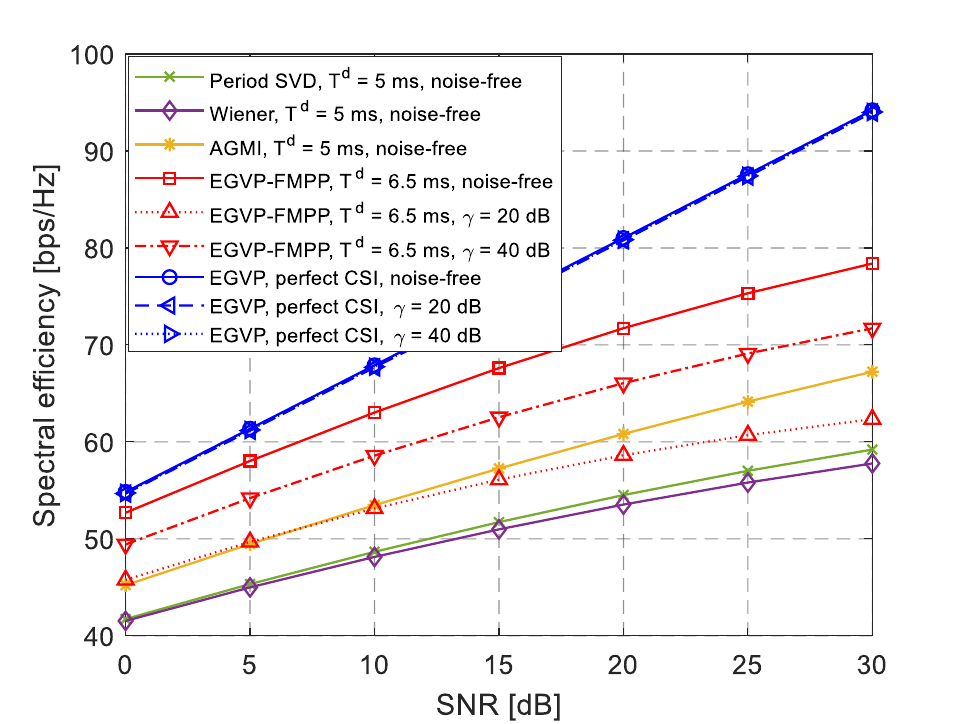}
\caption{SE performances with noisy channel samples, $v$ = 30 km/h, $\gamma$  = {20 dB, 40 dB}, $L_{\rm{ce}}=3$.}
\vspace{-0.3cm}
\label{fig_simulation_noise}
\end{figure}
In the end, we should note that the above results are based on the assumption of the noise-free channel samples. In reality, noisy sampling is inevitable and no doubt deteriorates the performances of our schemes. Many classic noise-canceling methods can be applied in this scenario, such as Prony \cite{1984PronyNoise}, minimum decription length (MDL) \cite{2009mdlNoise}, total linear squares (TLS) \cite{1987TLS} and Wiener filtering \cite{2006wienerNoise}. The detailed discussion is beyond the scope of our paper and the MDL method \cite{2009mdlNoise} is applied here to detect signals and suppress noise. Given two different channel sampling SNR $\gamma$, the SE performances of our schemes are illustrated in Fig. \ref{fig_simulation_noise}. It can be concluded that the channel sampling noise has limited impact on the EGVP scheme with  perfect CSI. However, it deteriorates the SE of the EGVP-FMPP scheme, especially when the noise power is large. Fortunately, the channel sampling SNR $\gamma$ = 40 dB is sufficient for the EGVP-FMPP scheme to outperform the benchmarks with noise-free sampling.

\section{Conclusion}\label{sec7}
This paper proposed a precoding matrix prediction method called EGVP in massive MIMO systems with mobility. Based on several periodic eigenvector samples of the downlink channels, the precoders at the BS were predicted by a polynomial-complexity channel weight interpolation and channel prediction. It was enabled by the decomposition of the precoders into a linear combinations of the CSI and the corresponding channel weights. The channel weights were interpolated by a complex exponential model. Considering the CSI delay, we further proposed a EGVP-FMPP method with an additional low-complexity channel prediction procedure. The asymptotic analyses proved that an error-free prediction can be achieved in our proposed EGVP schemes. The complexity advantage of our scheme over the traditional channel prediction schemes and the other precoding matrix prediction benchmarks were analyzed. Based on the industrial channel model, our schemes showed substantial improvement in terms of SE with low complexity in mobile environments.


%

\appendices

\section{Proof of Proposition \ref{Proposition eigen decomposition}}\label{Appendix Proposition eigen decompose}
The eigenvectors ${\overline {\bf{u}} _m}\left( t \right)$ obtained from the SVD of the channel matrix ${\bf{H}}\left( t \right)$ equal to the eigenvectors of ${\bf{H}}\left( t \right){\bf{H}}{\left( t \right)^T}$
\begin{small}
    \begin{equation}\label{eigen equation equivalent}
      \left( {\sum\limits_{j = 1}^M {{{\bf{h}}_j}\left( t \right){{\bf{h}}_j}{{\left( t \right)}^H}} } \right) \cdot {\overline {\bf{u}} _m}\left( t \right) = {\lambda _m}\left( t \right){\overline {\bf{u}} _m}\left( t \right).
  \end{equation}  
\end{small}
The left hand side of the above equation can be written as 
\begin{small}
    \begin{equation}\label{LHS eigen equation equivalent}
    \left( {\sum\limits_{j = 1}^M {{{\bf{h}}_j}\left( t \right){{\bf{h}}_j}{{\left( t \right)}^H}} } \right) {\overline {\bf{u}} _m}\left( t \right) = \sum\limits_{j = 1}^M {{{\bf{h}}_j}\left( t \right)\sum\limits_{n = 1}^{{N_t}{N_f}} {{h_{n,j}}\left( t \right){{\bar u}_{n,m}}\left( t \right)} } ,
    \nonumber
\end{equation}
\end{small}
where ${{h_{n,j}}\left( t \right)}$ and ${{\overline u }_{n,m}} \left( t \right)$ are the $n$-th element of ${{{\bf{h}}_j}\left( t \right)}$ and ${\overline {\bf{u}} _m}\left( t \right)$, respectively. Denote the channel weight as 
\begin{equation}\label{weight defination}
    {a_{m,j}}\left( t \right) = \frac{1}{{{\chi _m}\left( t \right)}}\sum\limits_{n = 1}^{{N_tN_f}} {{h_{n,j}}\left( t \right){{\overline u }_{n,m}} \left( t \right)} .
\end{equation}
Substitute the weight \eqref{weight defination} into the equation \eqref{eigen equation equivalent} and obtain the following equation
\begin{equation}\label{eigen decompostion final}
    {\overline{\bf{u}}_m}\left( t \right) = \sum\limits_{j = 1}^M {{a_{m,j}}\left( t \right){{\bf{h}}_j}} \left( t \right).
\end{equation}
Thus, Proposition \ref{Proposition eigen decomposition} is proved.
\section{Proof for Lemma \ref{Lemma of relevance model}}\label{Appendix for lemma of relevance model}
Based on the channel sparsity in the angle-delay domain, the $n$-th amplitude of the vector ${{\bf{g}}_m}\left( t \right)$, i.e., ${{\eta _{m,n}}\left( t \right)}$ can be approximately fitted by a complex exponential model as mentioned in \cite{2020yinMobility,2022ziao}
\begin{equation}\label{complex amplitude approximation model}
    {\eta _{m,n}}\left( t \right) \approx  \sum\limits_{q = 1}^{Q_{m,n}} {{c_{m,n,q}}{e^{j{\omega _{m,n,q}}t}}}, 
\end{equation}
where $Q_{m,n}$ is the number of exponentials and ${c_{m,n,q}}$ is the complex amplitude corresponding to the exponential ${e^{j{\omega _{m,n,q}}t}}$. The DFT leakage problem of the projection matrix ${\bf{W}}\left( {{N_f},{N_t}} \right)$ limits the approximation accuracy. Each angle-delay vector maps a corresponding angle-delay signature ${\bf{d}}_m$
\begin{equation}\label{asymptotic path analysis}
    \mathop {\lim }\limits_{{N_t}{N_f} \to \infty } {{\bf{w}}_n} = \mathop {\lim }\limits_{{N_t}{N_f} \to \infty } \frac{{{{\bf{d}}_m}\left( {{\theta _p},{\varphi _p},{\tau _p}} \right)}}{{\sqrt {{N_t}{N_f}} }} ,\forall p \in {{\cal P}},n \in {{\cal N}},
\end{equation}
where the set $\cal{P}$ and $\cal{N}$ are the multi-path index set and the angle-delay vector index set, respectively. Moreover, in asymptotic condition, each angular frequency ${{\omega _{m,n}}}$ corresponds to the Doppler frequency ${{\omega _{m,p}}}$ of the $p$-th path, i.e, $Q_{m,n}=1$ \cite{2022ziao}. This means that the complex amplitude $c_{m,n}$ maps the large scale channel parameter $\beta_{m,p}$ of the $p$-th path. Therefore, the asymptotic approximation model of \eqref{complex amplitude approximation model} is
\begin{equation}\label{asmptotic complex amplitude model}
    {\eta _{m,n}}\left( t \right) = {c_{m,n}}{e^{j{\omega _{m,n}}t}} \buildrel \Delta \over = \frac{{{\beta _{m,p}}{e^{j{\omega _{m,p}}t}}}}{{\sqrt {{N_t}{N_f}} }}. 
\end{equation}

Based on the channel in angle-delay domain equation \eqref{wideband band in angle-delay path form}, the channel inner product ${s_{i,j}}\left( t \right)$ satisfies 
\begin{equation}\label{relevance result}
\begin{array}{l}
{s_{i,j}}\left( t \right) = \left\langle {{{\bf{h}}_i}\left( t \right),{{\bf{h}}_j}\left( t \right)} \right\rangle \\
 = \left( {\sum\limits_{n = 1}^{{N_t}{N_f}} {{\eta _{j,n}}{{\left( t \right)}^H}{{\bf{w}}_n}^H} } \right)\left( {\sum\limits_{n = 1}^{{N_t}{N_f}} {{\eta _{j,n}}\left( t \right){{\bf{w}}_n}} } \right)\\
 = \sum\limits_{n = 1}^{{N_t}{N_f}} {{\eta _{j,n}}{{\left( t \right)}^H}{\eta _{i,n}}\left( t \right)} 
\end{array}.
\end{equation}
Substitute ${\eta _{m,n}}\left( t \right)$ with the equation \eqref{asmptotic complex amplitude model}, then the inner product ${s_{i,j}}\left( t \right)$ is
\begin{equation}\label{aysmptotic relevance model}
{s_{i,j}}\left( t \right) = \sum\limits_{n = 1}^{{N_f}{N_t}} {{c_{j,n}}^H{e^{ - j{\omega _{j,n}}t}}{c_{i,n}}{e^{j{\omega _{i,n}}t}} = \sum\limits_{x \in {{\cal{X}}_{i,j}}} {\rho _{i,j}^{\left( x \right)}{e^{j\omega _{i,j}^{\left( x \right)}t}}} ,} 
\nonumber
\end{equation}
where the set ${{\cal{X}}_{i,j}}$ denotes the indices collection of non-zero amplitudes. The size of ${{\cal{X}}_{i,j}}$ is denoted by $\left| {{\cal X}_{i,j}} \right| = {N_{{\rm{no}},i,j}} \le {N_f}{N_t}$. The non-zero complex amplitude is ${\rho _{i,j}^{\left( x \right)}}$ and the corresponding exponential is ${{e^{j\omega _{i,j}^{\left( x \right)}t}}}$. 

Therefore, Lemma \ref{Lemma of relevance model} is proved.
\section{Proof of Theorem \ref{theorem weight prediction model}}\label{Appendix Theroem weight predicion model}
Based on the channel weight definition \eqref{weight defination}, we rewrite it 
\begin{equation}\label{weight defination equal}
    {a_{m,j}}\left( t \right) = \frac{{\left\langle {{\overline{\bf{u}}_m}\left( t \right),{{\bf{h}}_j}\left( t \right)} \right\rangle }}{{{\chi _m}\left( t \right)}} = \frac{{{{\bf{h}}_j}{{\left( t \right)}^H}\sum\limits_{i = 1}^M {{a_{m,i}}\left( t \right){{\bf{h}}_i}\left( t \right)} }}{{{\chi _m}\left( t \right)}}.
\end{equation}
Substitute the channel ${{{\bf{h}}_j}\left( t \right)}$ in \eqref{weight defination equal} with \eqref{wideband band in angle-delay path form} 
\begin{equation}\label{weight defination equal 2}
    \begin{array}{l}
{a_{m,j}}\left( t \right) = \frac{{\sum\limits_{i = 1}^M {\left( {\sum\limits_{n = 1}^{{N_t}{N_f}} {{\eta _{j,n}}{{\left( t \right)}^H}{\eta _{i,n}}\left( t \right)} } \right){a_{m,i}}\left( t \right)} }}{{{\chi _m}\left( t \right)}}\\
 = \frac{{\sum\limits_{i = 1}^M {{s_{i,j}}\left( t \right){a_{m,i}}\left( t \right)} }}{{{\chi _m}\left( t \right)}}
\end{array}.
\end{equation}
Then, rewrite the above equation in a matrix form
\begin{equation}\label{relevance eigen equation}
    {\bf{S}}\left( t \right)\left[ {\begin{array}{*{20}{c}}
{{a_{1,m}}\left( t \right)}\\
{{a_{2,m}}\left( t \right)}\\
 \cdots \\
{{a_{M,m}}\left( t \right)}
\end{array}} \right] = {\chi _m}\left( t \right)\left[ {\begin{array}{*{20}{c}}
{{a_{1,m}}\left( t \right)}\\
{{a_{2,m}}\left( t \right)}\\
 \cdots \\
{{a_{M,m}}\left( t \right)}
\end{array}} \right],
\end{equation}
where ${\bf{S}}\left( t \right)$ is the collection of all channel inner products
\begin{equation}\label{relevane matrix definition}
    {\bf{S}}\left( t \right) = \left[ {\begin{array}{*{20}{c}}
{{s_{1,1}}\left( t \right)}&{{s_{1,2}}\left( t \right)}& \cdots &{{s_{1,M}}\left( t \right)}\\
{{s_{2,1}}\left( t \right)}&{{s_{2,2}}\left( t \right)}& \cdots &{{s_{2,M}}\left( t \right)}\\
 \cdots & \cdots & \cdots & \cdots \\
{{s_{M,1}}\left( t \right)}&{{s_{M,2}}\left( t \right)}& \cdots &{{s_{M,M}}\left( t \right)}
\end{array}} \right],
\end{equation}
where the element ${s_{i,j}}\left( t \right)$ is defined in \eqref{channel inner product DFT}. Obviously, the equation \eqref{relevance eigen equation} is the $m$-th eigen equation of matrix ${\bf{S}}\left( t \right)$. It shares the same eigenvalue $\chi_m\left(t\right)$ with the matrix ${\bf{ H}}\left( t \right)$. At the same time, the corresponding eigenvector is the collection of the channel weights
\begin{equation}\label{weigth vector definition}
    {{\bf{a}}_m}\left( t \right) = {\left[ {\begin{array}{*{20}{c}}
{{a_{1,m}}\left( t \right)}&{{a_{2,m}}\left( t \right)}& \cdots &{{a_{M,m}}\left( t \right)}
\end{array}} \right]^T}.
\end{equation}
Due to the definition of ${s_{i,j}}\left( t \right)$ in \eqref{channel inner product DFT}, the matrix ${\bf{S}}\left( t \right)$ is a Hermitian matrix. And the eigenvalue $\chi_m\left(t\right)$ is non-zero. Therefore, the matrix ${\bf{S}}\left( t \right)$ is definite and non-singular.

The eigenvalue $\chi_m\left(t\right)$ is obtained by solving 
\begin{equation}\label{solving eigen equation}
   \left( {{\bf{S}}\left( t \right) - {\chi _m}\left( t \right){{\bf{I}}_M}} \right){{\bf{a}}_m}\left( t \right) = \bf{0},
\end{equation}
where ${{\bf{I}}_M}$ is a $M$-dimension unit matrix. 

We first let $M=2$ and focus on obtaining the first eigenvector ${{\bf{a}}_1}\left( t \right)$. The other channel weight solution  ${{\bf{a}}_2}\left( t \right)$ can be similarly proved. Rewrite the equation \eqref{solving eigen equation} as
\begin{equation}\label{sovling eigen equation layer two}
   \left[ {\begin{array}{*{20}{c}}
{{s_{1,1}}\left( t \right) - {\chi _1}\left( t \right)}&{{s_{1,2}}\left( t \right)}\\
{{s_{2,1}}\left( t \right)}&{{s_{2,2}}\left( t \right) - {\chi _1}\left( t \right)}
\end{array}} \right]{{\bf{a}}_m}\left( t \right) = \bf{0}.
\end{equation}
Using Gaussian elimination, turn the the left hand side matrix in \eqref{sovling eigen equation layer two} into an upper triangular matrix $\widetilde {\bf{S}}\left( t \right)$
\begin{equation}\label{sovling eigen equation layer two equals}
    \left[ {\begin{array}{*{20}{c}}
{{{{ S}}_{1,1}}\left( t \right)}&{{{{S}}_{1,2}}\left( t \right)}\\
0&{{{{S}}_{2,2}}\left( t \right)}
\end{array}} \right]{{\bf{a}}_m}\left( t \right) = {\bf{0}},
\end{equation}
where each element of $\widetilde {\bf{S}}\left( t \right)$ is denoted by
\begin{equation}\label{sovling eigen equation layer two upper triangular matrix}
\left\{ \begin{array}{l}
{{{S}}_{1,1}}\left( t \right) = {s_{1,1}}\left( t \right) - {\chi _1}\left( t \right),\\
{{{S}}_{1,2}}\left( t \right) = {s_{1,2}}\left( t \right),\\
{{{S}}_{2,2}}\left( t \right) = {s_{2,2}}\left( t \right) - {\chi _1}\left( t \right) - \frac{{{{\left| {{{{S}}_{1,2}}\left( t \right)} \right|}^2}}}{{{{{S}}_{1,1}}\left( t \right)}}.
\end{array} \right.
\end{equation}
As ${\bf{S}}\left( t \right)$ is non-singular, the rank of the matrix $\widetilde {\bf{S}}\left( t \right)$ is 1. Therefore, the eigen equation \eqref{sovling eigen equation layer two equals} satisfies
\begin{equation}\label{sovling eigen equation layer two results}
\left\{ \begin{array}{l}
{a_{1,1}}\left( t \right){{{S}}_{1,1}}\left( t \right) + {a_{2,1}}\left( t \right){{{S}}_{1,2}}\left( t \right) = 0,\\
{{{S}}_{2,2}}\left( t \right) = 0,{{{S}}_{1,2}}\left( t \right) \ne 0,
\end{array} \right.
\end{equation}
which equals to
\begin{equation}\label{sovling eigen equation layer two final results one }
   {a_{1,1}}\left( t \right) =  - \frac{{{s_{1,2}}\left( t \right)}}{{{s_{1,1}}\left( t \right) - {\chi _1}\left( t \right)}}{a_{2,1}}\left( t \right), 
\end{equation}
\begin{equation}\label{sovling eigen equation layer two final results two}
   \left( {{s_{2,2}}\left( t \right) - {\chi _1}\left( t \right)} \right)\left( {{s_{1,1}}\left( t \right) - {\chi _1}\left( t \right)} \right) = {\left| {{s_{1,2}}\left( t \right)} \right|^2}.
\end{equation}
Solve the equation \eqref{sovling eigen equation layer two final results two} and obatin the eigenvalue 
\begin{small}
\begin{equation}\label{sovling eigen equation layer two eigenvalue}
{\chi _1}\left( t \right){\rm{ = }}\frac{{{s_{1,1}}\left( t \right) + {s_{2,2}}\left( t \right) \pm \sqrt {{{\left( {{s_{1,1}}\left( t \right) - {s_{2,2}}\left( t \right)} \right)}^2} + 4{{\left| {{s_{1,2}}\left( t \right)} \right|}^2}} }}{2}.
\end{equation}
\end{small}
As the first eigenvalue $\chi_1\left(t\right)$ is the greater one, we take the greater value of \eqref{sovling eigen equation layer two eigenvalue}. Then the equation \eqref{sovling eigen equation layer two final results one } becomes
\begin{small}
 \begin{equation}\label{sovling eigen equation layer two eigen vector relation}
\left\{ \begin{array}{l}
{a_{1,1}}\left( t \right) = {{{\cal K}}_{1,2}}\left( t \right){a_{2,1}}\left( t \right),\\
{{{\cal K}}_{1,2}}\left( t \right) = \frac{{2{s_{1,2}}\left( t \right)}}{{\sqrt {{{\left( {{s_{1,1}}\left( t \right) - {s_{2,2}}\left( t \right)} \right)}^2} + 4{{\left| {{s_{1,2}}\left( t \right)} \right|}^2}}  + \left( {{s_{2,2}}\left( t \right) - {s_{1,1}}\left( t \right)} \right)}},
\end{array} \right.
\end{equation}   
\end{small}
where ${\cal K}_{1,2}\left(t\right)$ is the ratio of ${a_{1,1}}\left( t \right)$ to ${a_{2,1}}\left( t \right)$. 

Then, we analyze the asymptotic value of ${\cal K}_{1,2}\left(t\right)$. Let the normalized denominator of ${\cal K}_{1,2}\left(t\right)$ be 

\begin{equation}\label{sovling eigen equation layer two ratio denominator}
    \begin{array}{l}
{D_{1,2}}\left( t \right) = \sqrt {\frac{{{{\left( {{s_{1,1}}\left( t \right) - {s_{2,2}}\left( t \right)} \right)}^2} + 4{{\left| {{s_{1,2}}\left( t \right)} \right|}^2}}}{{{N_t}^2{N_f}^2}}} \\
 + \frac{{\left( {{s_{2,2}}\left( t \right) - {s_{1,1}}\left( t \right)} \right)}}{{{N_t}{N_f}}}.
\end{array}
\end{equation}  

Based on Assumption \ref{assumption for asymptotic analysis}, we have 
\begin{equation}\label{sovling eigen equation layer two asympototic ratio}
  \begin{array}{l}
\mathop {\lim }\limits_{{N_t}{N_f} \to \infty } {{\cal{K}}_{1,2}}\left( t \right) = \frac{{\frac{{{s_{1,2}}\left( t \right)}}{{{N_t}{N_f}}}}}{{{D_{1,2}}\left( t \right)}}\\
 = \frac{{\frac{{{s_{1,2}}\left( t \right)}}{{{N_t}{N_f}}}}}{{\sqrt {{{\left( {{\cal{A}} - {\cal{B}}} \right)}^2} + {{\cal{C}}^2}}  + {\cal{A}} - {\cal{B}}}}\\
\end{array}. 
\end{equation}
And the absolute value of ${\cal K}_{1,2}\left(t\right)$ converges to
\begin{small}
  \begin{equation}\label{abs of sovling eigen equation layer two asympototic ratio}
    \mathop {\lim }\limits_{{N_t}{N_f} \to \infty } \left| {{{\cal K}}_{1,2}\left( t \right)} \right| = \frac{{\cal{C}}}{{\sqrt {{{\left( {{\cal{A}} - {\cal{B}}} \right)}^2} + {{\cal{C}}^2}}  + {\cal{A}} - {\cal{B}}}}.
\end{equation}  
\end{small}
Learned from \eqref{sovling eigen equation layer two asympototic ratio} and \eqref{abs of sovling eigen equation layer two asympototic ratio}, ${\cal K}_{1,2}\left(t\right)$ can be modeled by ${s_{1,2}}\left( t \right)$ multiplied by a time-invariant scaling factor.

However, the eigenvector ${{\bf{a}}_1}\left( t \right)$ is non-unique. Without losing generality, we first consider a particular solution 
\begin{equation}\label{first channel weight value assumption}
\overline a_{2,1}\left( t \right) = {{\cal U}}{s_{1,2}}\left( t \right),   
\end{equation}
where $\cal U$ is a time-invariant scaling factor. Learned from \eqref{sovling eigen equation layer two eigen vector relation} and Lemma \ref{Lemma of relevance model}, the particular solution $\overline a_{1,1}\left( t \right)$ is
\begin{equation}\label{second channel weight value}
    \begin{array}{l}
\mathop {\lim }\limits_{{N_t}{N_f} \to \infty } {\overline a_{1,1}}\left( t \right) = {{\cal{K}}_{1,2}}\left( t \right){\cal{U}}{s_{1,2}}\left( t \right)\\
 = \widetilde {{\cal U}}{s_{1,2}}{\left( t \right)^2} \buildrel \Delta \over = \sum\limits_{l = 1}^{L_{i,j}} {b_{1,1}^{\left( l \right)}{e^{j\omega _{1,1}^{\left( l \right)}t}}} 
\end{array}, 
\end{equation}
where, $\widetilde {{\cal U}}$ is a time-invariant scaling factor. Thus, the particular solution of channel weight ${{\bf{\bar a}}_{1,1}}\left( t \right)$ can be modeled by a complex exponential model associated with ${s_{1,2}}\left( t \right)$.

Likewise, the general solution of the channel weight $a_{2,1}\left( t \right)$ is also estimated by a complex exponential model
\begin{equation}\label{first channel weight value general assumption}
    \mathop {\lim }\limits_{{N_t}{N_f} \to \infty }{a_{2,1}}\left( t \right) \buildrel \Delta \over = \sum\limits_{l = 1}^{L_{i,j}} {b_{2,1}^{\left( l \right)}{e^{j\omega _{2,1}^{\left( l \right)}t}}}.
\end{equation}
According to \eqref{sovling eigen equation layer two eigen vector relation}, both $a_{2,1}\left( t \right)$ and $a_{1,1}\left( t \right)$ can also be modeled by a complex exponential model which may be irrelevant with ${s_{1,2}}\left( t \right)$. 

Even though the above derivation is made when $M=2$, this limitation can be relaxed to $M>2$. In this case, an iterative form of channel weight is also proven to be asymptotically estimated by a complex exponential model.

When $M>2$,
perform an elementary row operation and turn the matrix ${{\bf{S}}\left( t \right) - {\chi _1}\left( t \right){{\bf{I}}_M}}$  to an upper triangular matrix $\widetilde{\bf{S}}\left(t\right)$. First, define an iterative matrix $\widetilde{\bf{S}}^I\left(t\right),I\in\left\{0,1,2,\cdots,M-1\right\}$ as the result of the $I$-th elementary row operation of ${{\bf{S}}\left( t \right) - {\chi _1}\left( t \right){{\bf{I}}_M}}$. And the initial matrix is $\widetilde{\bf{S}}^{0}\left(t\right)={{\bf{S}}\left( t \right) - {\chi _1}\left( t \right){{\bf{I}}_M}}$. Obviously, the last iterative matrix $\widetilde{\bf{S}}^{M-1}\left(t\right)$ equals to the matrix $\widetilde{\bf{S}}\left(t\right)$
\begin{small}
\begin{equation}\label{general case matrix of upper triangular matrix}
\left[ {\begin{array}{*{20}{c}}
{S_{1,1}^{M{\rm{ - }}1}\left( t \right)}&{S_{1,2}^{M{\rm{ - }}1}\left( t \right)}& \cdots &{S_{1,M}^{M{\rm{ - }}1}\left( t \right)}\\
0&{S_{2,2}^{M{\rm{ - }}1}\left( t \right)}& \cdots &{S_{2,M}^{M{\rm{ - }}1}\left( t \right)}\\
 \cdots & \cdots & \cdots & \cdots \\
0& \cdots &{S_{M{\rm{ - }}1,M{\rm{ - }}1}^{M{\rm{ - }}1}\left( t \right)}&{S_{M{\rm{ - }}1,M}^{M{\rm{ - }}1}\left( t \right)}\\
0&0& \cdots &{S_{M,M}^{M{\rm{ - }}1}\left( t \right)}
\end{array}} \right].
\end{equation}    
\end{small}

Then, we aim to derive the iterative form of the elements of \eqref{general case matrix of upper triangular matrix} instead of the closed-form like  \eqref{sovling eigen equation layer two eigen vector relation}. The elements of the initial matrix $S_{i,j}^0\left( t \right)$ satisfy
\begin{equation}\label{inital elements of upper triangular matrix iteration 1}
   S_{i,j}^0\left( t \right) = \left\{ \begin{array}{l}
{s_{i,j}}\left( t \right),i \ne j,\\
{s_{i,j}}\left( t \right) - {\chi _1}\left( t \right),i = j.
\end{array} \right.
\end{equation}
When $I=1$, $\widetilde{\bf{S}}^{1}\left(t\right)$ becomes
\begin{small}
 \begin{equation}\label{general case matrix of upper triangular matrix iteration 1}
    \left[ {\begin{array}{*{20}{c}}
{S_{1,1}^1\left( t \right)}&{S_{1,2}^1\left( t \right)}& \cdots &{S_{1,M}^1\left( t \right)}\\
0&{S_{2,2}^1\left( t \right)}& \cdots &{S_{2,M}^1\left( t \right)}\\
 \cdots & \cdots & \cdots & \cdots \\
0& \cdots &{S_{M{\rm{ - }}1,M{\rm{ - }}1}^1\left( t \right)}&{S_{M{\rm{ - }}1,M}^1\left( t \right)}\\
0&{S_{M,2}^1\left( t \right)}& \cdots &{{S_{M,M}}\left( t \right)}
\end{array}} \right],
\end{equation}   
\end{small}
where the elements satisfy
\begin{equation}\label{elements of upper triangular matrix iteration 1}
S_{i,j}^1\left( t \right) = \left\{ \begin{array}{l}
S_{i,j}^0\left( t \right),i \le 1,\\
S_{i,j}^0\left( t \right) - \frac{{S_{i,1}^0\left( t \right)}}{{S_{1,1}^0\left( t \right)}}S_{1,j}^0\left( t \right),1 < i \le j,\\
0,i > j.
\end{array} \right.
\end{equation}
The indices $i$ and $j$ denote the $i$-th row and $j$-column of $\widetilde{\bf{S}}\left(t\right)$, respectively. 
When $I\ge1$, the general form of the elements of $\widetilde{\bf{S}}^{I}\left(t\right)$ are
\begin{small}
    \begin{equation}\label{general case elements of upper triangular matrix iteration}
S_{i,j}^I\left( t \right) = \left\{ \begin{array}{l}
S_{i,j}^{I - 1}\left( t \right),i \le I,\\
S_{i,j}^{I - 1}\left( t \right) - \frac{{S_{i,I}^{I - 1}\left( t \right)}}{{S_{I,I}^{I - 1}\left( t \right)}}S_{I,j}^{I - 1}\left( t \right),i > I,j > I,\\
0,j \le I < i.
\end{array} \right.
\end{equation}
\end{small}
Based on the iterative equation \eqref{general case elements of upper triangular matrix iteration} and initial condition \eqref{inital elements of upper triangular matrix iteration 1}, the eigen equation  $\widetilde{\bf{S}}\left( t \right){{\bf{a}}_m}\left( t \right) = {\bf{0}}$ is equally rewritten as
\begin{equation}\label{gerneral sovlution of upper triangular matrix}
\sum\limits_{i = n}^M {S_{n,i}^{M - 1}\left( t \right){a_{i,m}}\left( t \right) = 0,n \in 1 \cdots M}, 
\end{equation}
where $n$ denotes the $n$-th row of $\widetilde{\bf{S}}\left( t \right)$. The solution to the series of equations \eqref{gerneral sovlution of upper triangular matrix} is   
\begin{equation}\label{gerneral iterative results of weight}
    {a_{n,m}}\left( t \right) =  - \sum\limits_{i = n}^{M - 1} {{{\cal K}}_{n,i + 1}^{M - 1}\left( t \right){a_{i + 1,m}}\left( t \right),n \in 1 \cdots M - 1} ,   
\end{equation}
where the ratio is defined by ${{\cal K}}_{i,j}^{M - 1}\left( t \right) = \frac{{S_{i,j}^{M-1}\left( t \right)}}{{S_{i,i}^{M-1}\left( t \right)}}$.
According to Appendix \ref{Appendix Theroem weight predicion model}, the solution ${{\bf{a}}_m}\left( t \right)$ is non-unique and ${a_{M,m}}\left( t \right)$ can be any possible value. 

The ratio ${{\cal K}}_{i,j}^{I - 1}\left( t \right)$ is first analyzed in order to prove the iterative property of the channel weight $a_{n,m}\left(t\right)$. We have proved that the ratio ${{{\cal K}}_{1,2}}\left( t \right) = \frac{{{S_{1,2}}\left( t \right)}}{{{S_{1,1}}\left( t \right)}}$ can be modeled by ${s_{1,2}}\left( t \right)$ multiplied by a time-invariant scaling factor. Likewise, when $I=1$, the ratio
${{\cal K}}_{i,j}^0\left( t \right) = \frac{{S_{i,j}^0\left( t \right)}}{{S_{i,i}^0\left( t \right)}}$ shares the same property. When $I>1$, 
based on the iterative equation
\begin{equation}\label{general case iterative equation of S}
    S_{i,j}^I\left( t \right) = S_{i,j}^{I - 1}\left( t \right) - {{\cal K}}_{I,j}^{I - 1}\left( t \right)S_{i,I}^{I - 1}\left( t \right),i > I,j > I,
\end{equation}
the next iteration $S_{i,j}^1\left( t \right)$ can be modeled by a complex exponential model like $S_{i,j}^0\left( t \right)$. As a result, the ratio ${\cal K}_{i,j}^1\left( t \right) = \frac{{S_{i,j}^1\left( t \right)}}{{S_{i,i}^1\left( t \right)}}$ can be also estimated by a complex exponential model, which means that all ratios ${{\cal K}}_{i,j}^{I - 1}\left( t \right)$ hold this conclusion.

In the end, a series of solutions ${{\bf{a}} _m}\left( t \right)$ can be obtained by solving \eqref{gerneral iterative results of weight} given the particular solution ${{\overline{a}}_{M,m}}\left( t \right) = {{\cal U}}{s_{M,m}}\left( t \right)$ or the general solution ${{a}_{M,m}}\left( t \right)$. Therefore, the weight ${a_{M,m}}\left( t \right)$ can be also modeled by a complex exponential model when $M>2$. 
Based on the above discussion about cases when $M=2$ and $M>2$, Theorem \ref{theorem weight prediction model} is proved. 
\section{Proof of Theorem \ref{theorem for asmptotic channel estimation}}\label{Appendix theorem for asmptotic channel estimation}
Given ${N_{{\rm{ce}}}} = 2$, the prediction order of FMPP is $L_{\rm{ce}}=1$. When ${N_t},{N_f} \to \infty$, Appendix \ref{Appendix for lemma of relevance model} has proved that each exponential ${{e^{j{\omega _{m,n}}}}}$ of the complex amplitude ${\eta _{m,n}}\left( t \right)$ corresponds to the Doppler frequency shift ${{e^{j{\omega _{m,p}}}}}$ of the channel $ {\bf{h}}_m\left( t \right)$. Based on \eqref{asmptotic complex amplitude model}, the asymptotic channel PE is 
\begin{small}
  \begin{equation}\label{asmptotic channel prediction error}
        \begin{array}{l}
        \mathop {\lim }\limits_{{N_t}{N_f} \to \infty }  \mathbb{E}\left\{ {\frac{{\left\| {{{\bf{h}}_m}\left( t \right) - {{\hat {\bf{h}}}_m}\left( t \right)} \right\|_2^2}}{{\left\| {{{\bf{h}}_m}\left( t \right)} \right\|_2^2}}} \right\}\\
         = \mathbb{E}\left\{ {\frac{{\left\| {\sum\limits_{p = 1}^P {{\beta _{p,m}}{e^{j{\omega _{p,m}}t}}{{\bf{d}}_m}}  - \sum\limits_{n = 1}^{{N_t}{N_f}} {{c_{m,n}}{e^{j{\omega _{m,n}}t}}{{\bf{w}}_n}} } \right\|_2^2}}{{\left\| {{{\bf{h}}_m}\left( t \right)} \right\|_2^2}}} \right\}\\
         = \mathbb{E}\left\{ {\frac{{\left\| {\sum\limits_{n \in {{{\cal N}}_{{\rm{nz}}}}} {\left( {\frac{{{\beta _{m,n}}{e^{j{\omega _{m,n}}t}}}}{{\sqrt {{N_f}{N_t}} }} - {c_{m,n}}{e^{j{\omega _{m,n}}t}}} \right){{\bf{w}}_n}} } \right\|_2^2}}{{\left\| {{{\bf{h}}_m}\left( t \right)} \right\|_2^2}}} \right\} = 0,
        \end{array}
    \end{equation}  
\end{small}    
where ${{{{\cal N}}_{{\rm{nz}}}}}$ is the collection of non-zero indices of the angle-delay vector ${\bf{w}}_n$. Thus, Theorem \ref{theorem for asmptotic channel estimation} is proved.
\section{Proof of Theorem \ref{theorem for asmptotic analysis of EGVP}}\label{Appendix theorem for asmptotic eigen prediction}
The eigenvector ${\overline {\bf{u}}_1 \left( t \right)}$ corresponding to the maximum eigenvalue is analyzed first. 
Based on \eqref{eigenvector decomposition} and \eqref{eigen vector reconstruction}, the eigenvector PE is equally rewritten as
    \begin{equation}\label{asymptotic prediction error equation}
        \begin{array}{l}
        \overline {\bf{u}}_1 \left( t \right) - {\hat {\bf{u}}_1}\left( t \right) = \sum\limits_{j = 1}^M {{{\bf{h}}_j}\left( t \right)\left( {{a_{1,j}}\left( t \right) - {{\hat a}_{1,j}}\left( t \right)} \right)} \\
         = {\bf{H}}{\left( t \right)^T}\left( {{{\bf{a}}_1}\left( t \right) - {{\hat {\bf{a}}}_1}\left( t \right)} \right).
        \end{array}
    \end{equation}
     The particular solution of the predicted channel weight is given in Appendix \ref{Appendix Theroem weight predicion model}. The closed form of the channel weight is first \eqref{first channel weight value assumption} proved, i.e., ${\hat {\bf{a}}_1}\left( t \right)  \buildrel \Delta \over =  {\overline {\bf{a}} _1}\left( t \right)$. Then the iterative form of the channel weight \eqref{gerneral iterative results of weight} can be similarly derived.
       
    According to Lemma \ref{Lemma of relevance model} and \eqref{channel inner product DFT}, the number of the exponentials in \eqref{weight prediction model} is $N_{{\rm{no}}}^{{\rm{max}}}$. Therefore, at least $2N_{\rm{no}}^{\rm{max}}$ samples ensure a successful prediction of the ${\hat {\bf{a}}_1}\left( t \right)$ 
    \begin{equation}\label{asymptotic weight prediction model}
        \begin{array}{l}
        \mathop {\lim }\limits_{{N_t}{N_f} \to \infty }{\hat {\bf{a}}_1}\left( t \right) = {\left[ {\begin{array}{*{20}{c}}
        {{{\hat a}_{1,1}}\left( t \right)}&{{{\hat a}_{1,2}}\left( t \right)}
        \end{array}} \right]^T}\\
        {\left[ {\begin{array}{*{20}{c}}
        {\sum\limits_{l = 1}^{N_{{\rm{no}}}^{\max }} {b_{1,1}^{\left( l \right)}{e^{j\omega _{1,1}^{\left( l \right)}t}}} }&{\sum\limits_{l = 1}^{N_{{\rm{no}}}^{\max }} {b_{1,2}^{\left( l \right)}{e^{j\omega _{1,2}^{\left( l \right)}t}}} }
        \end{array}} \right]^T}
        \end{array}.
    \end{equation}
    
    According to Theorem \ref{theorem weight prediction model}, each amplitude ${b_{1,2}^{\left( l \right)}}$ of ${{\hat a}_{1,2}}\left( t \right)$ maps the amplitude ${{\rho _{1,2,x}}}$ of ${s_{1,2}}\left( t \right)$. And the poles ${{e^{j\omega _{1,2}^{\left( l \right)}}}}$ of ${{\hat a}_{1,2}}\left( t \right)$ maps the poles ${e^{j{\omega _{1,2,x}}}}$ of ${s_{1,2}}\left( t \right)$. Lemma \ref{Lemma of relevance model} shows that $s_{1,2}\left (t\right)$ can be approximated by a $N_{\rm{no}}^{\rm{max}}$-order complex exponential model \eqref{channel inner product DFT}. Therefore, the interpolated channel weight satisfies $\mathop {\lim }\limits_{{N_t}{N_f}{ \to \infty }} {\hat a_{1,2}}\left( t \right){\rm{ }} = {{\cal U}}{s_{1,2}}\left( t \right)$. In terms of the estimation of ${{{\hat a}_{1,1}}\left( t \right)}$, based on Appendix \ref{Appendix for lemma of relevance model}, the poles ${{e^{j{\omega _{1,2,x}}}}}$ of $s_{1,2}\left (t\right)$ denotes the product of the non-orthogonal Doppler frequencies of channel ${\bf{h}}_1\left(t\right)$ and channel ${\bf{h}}_2\left(t\right)$. Thus the order of 
    ${\widetilde {{\cal U}}{s_{1,2}}{{\left( t \right)}^2}}$ is also $N_{\rm{no}}^{\rm{max}}$.
    Then the real channel weight ${\widetilde {{\cal U}}{s_{1,2}}{{\left( t \right)}^2}}$ can also be modeled by an $N_{\rm{no}}^{\rm{max}}$-order complex exponential model \eqref{channel inner product DFT}, which means that the estimated channel weight $\mathop {\lim }\limits_{{N_t}{N_f} \to \infty }{\hat a_{1,1}}\left( t \right) = \sum\limits_{l = 1}^{N_{{\rm{no}}}^{\max }} {b_{1,1}^{\left( l \right)}{e^{j\omega _{1,1}^{\left( l \right)}t}}}$ can map the real channel weight ${\widetilde {{\cal U}}{s_{1,2}}{{\left( t \right)}^2}}$. Therefore, the error-free channel weight estimation is achieved $\mathop {\lim }\limits_{{N_t}{N_f} \to \infty } {\hat {\bf{a}}_1}\left( t \right){\rm{ }} = {\overline {\bf{a}} _1}\left( t \right)$ as well as the eigenvector estimation $\mathop {\lim }\limits_{{N_t}{N_f} \to \infty } {\widehat {\bf{u}}_1}\left( t \right) = {\overline {\bf{u}} _1}\left( t \right)$. 
    
    Likewise, the above derivation is also valid for the other eigenvectors. In the end, Theorem \ref{theorem for asmptotic analysis of EGVP} is proved.



\ifCLASSOPTIONcaptionsoff
  \newpage
\fi



%
\bibliographystyle{IEEEtran}
\bibliography{IEEEabrv,reference}

\end{document}